\documentclass[acmsmall, screen, nonacm]{acmart}

\usepackage{xspace}
\usepackage{listings}
\usepackage{xcolor}
\usepackage{lstautogobble}
\usepackage{color}
\usepackage{amsmath}
\usepackage{todonotes}
\usepackage{mathpartir}
\usepackage{subcaption}
\usepackage{float}
\usepackage{stmaryrd}
\usepackage{csquotes}
\usepackage{enumitem}
\usepackage{hyperref}
\usepackage{amsthm}
\usepackage{cleveref}

\crefname{plemma}{lemma}{lemmas}
\title{Pacing Types: Safe Monitoring of Asynchronous Streams}

\author{Florian Kohn}
\email{florian.kohn@cispa.de}
\orcid{0000-0001-9672-2398}
\affiliation{
  \institution{CISPA Helmholtz Center for Information Security}
  \country{Germany}
}
\author{Arthur Correnson}
\email{arthur.correnson@cispa.de}
\orcid{0000-0003-2307-2296}
\affiliation{
  \institution{CISPA Helmholtz Center for Information Security}
  \country{Germany}
}
\author{Jan Baumeister}
\email{jan.baumeister@cispa.de}
\orcid{0000-0002-8891-7483}
\affiliation{
  \institution{CISPA Helmholtz Center for Information Security}
  \country{Germany}
}
\author{Bernd Finkbeiner}
\email{finkbeiner@cispa.de}
\orcid{0000-0002-4280-8441}
\affiliation{
  \institution{CISPA Helmholtz Center for Information Security}
  \country{Germany}
}

\newcommand{\CLola}{\textsc{StreamCore}\xspace}
\newcommand{\stream}[1]{{\footnotesize\texttt{#1}}\xspace}
\newcommand{\bad}{\lightning}
\newcommand{\headline}[1]{\subsubsection*{\textbf{#1}}}

\definecolor{bluekeywords}{rgb}{0.13, 0.13, 1}
\definecolor{greentypes}{rgb}{0, 0.5, 0}
\definecolor{orangecomments}{rgb}{1, 0.5, 0.1}
\definecolor{redstrings}{RGB}{171, 114, 2}
\definecolor{graynumbers}{rgb}{0.5, 0.5, 0.5}
\definecolor{goldcomments}{rgb}{0.6, 0.4, 0.08}

\lstdefinelanguage{Lola}{
  keywords=[0]{input, output, trigger, constant, import, spawn, eval, close, with, when},
  moredelim=**[is][\color{greentypes}@]{@}{@},
  keywordstyle=[0]\bfseries\color{bluekeywords},
  keywords=[1]{if, then, else, aggregate, defaults, offset, last, prev, by, or, to, sin, cos, abs, hold, over, using, over_instances},
  keywords=[2]{Variable, String, Int, Int64, UInt, UInt64, Bool, Float32, Float64, Float},
  keywordstyle=[2]\color{greentypes},
  sensitive=false,
  comment=[l]{//},
  morecomment=[s]{/*}{*/},
  morestring=[b]',
  morestring=[b]",
  literate={\\@}{@}1
}
\begin{abstract}
  Stream-based monitoring is a real-time safety assurance mechanism for complex cyber-physical systems such as unmanned aerial vehicles.
  In this context, a \emph{monitor} aggregates streams of input data from sensors and other sources to give real-time statistics and assessments of the system's health.
  Since monitors are safety-critical components, it is crucial to ensure that they are free of potential runtime errors.
  One of the central challenges in designing reliable stream-based monitors is to deal with the asynchronous nature of data streams:
  in concrete applications, the different sensors being monitored produce values at different speeds, and it is the monitor's responsibility to correctly react to the asynchronous arrival of different streams of values.
  To ease this process, modern frameworks for stream-based monitoring such as \textsc{RTLola} feature an expressive specification language that allows to finely specify data synchronization policies.
  While this feature dramatically simplifies the design of
  monitors, it can also lead to subtle runtime errors.
  To mitigate this issue, this paper presents \emph{pacing types}, a novel type system implemented in RTLola to ensure that monitors for asynchronous streams are well-behaved at runtime.
  We formalize the essence of pacing types for a core fragment of \textsc{RTLola}, and present a soundness proof of the pacing type system using a new logical relation.
\end{abstract}

\begin{document}
  \maketitle
  \section{Introduction}

Runtime verification is an approach to reliable computing in which systems are monitored during their execution. 
% This is in opposition to static methods which are analyzing a system prior to its execution to precisely determine whether it is guaranteed to always execute as expected.
Runtime verification is particularly useful for increasing the reliability of complex cyber-physical systems (CPS) for which obtaining a model for static verification is difficult or impossible.

Stream-based monitoring is an approach to runtime verification in which a monitor
collects data through input streams, which are then combined and aggregated in output streams to give real-time assessments of a system's health.
\Cref{fig:stream-based-monitoring} illustrates how a stream-based monitor integrates with a CPS.
The light grey dots on the left represent input stream values measured by system sensors.
The dark grey dots indicate computed output stream values, which are used to determine a boolean verdict.

\begin{figure}
    \centering
        \includegraphics[width=0.75\linewidth]{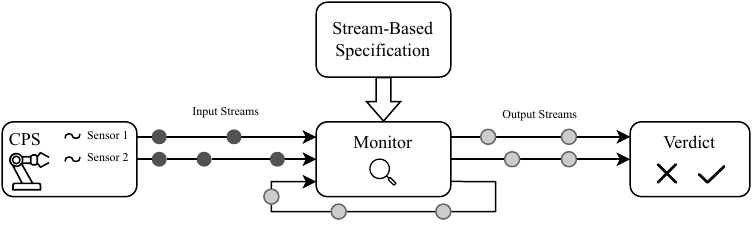}
        \caption{An overview of Stream-based Monitoring}
        \label{fig:stream-based-monitoring}
\end{figure}

There exist many formalisms and programming languages to describe stream-based monitors \cite{d2005lola, DBLP:conf/cav/FaymonvilleFSSS19, volostream, DBLP:conf/rv/KallwiesLSSTW22,DBLP:conf/fm/GorostiagaS21,DBLP:conf/fm/PerezGD24}.
\textsc{RTLola}, for example, is a stream-based runtime monitoring framework based on the
formalism of stream equations. In \textsc{RTLola}, a monitor is specified by listing
equations relating input and output streams. Certain boolean output streams, called triggers,
indicate whether the system being monitored is behaving as expected.

From a system perspective, stream-based monitors are passive components.
They react to inputs from the system as they arrive.
If different system sensors generate inputs at varying rates, the monitor must handle the asynchronous arrival of data.
To address this issue, \textsc{RTLola} features a system of \emph{pacing annotations} where
each output stream is labeled with additional information indicating when a stream must
produce a fresh value in reaction to the asynchronous arrival of new input values.
While the pacing annotation system of \textsc{RTLola} has proven to be a
useful in practical applications \cite{Lola2.0,baumeister2020rtlola,volostream}, it also introduces an critical challenge: ensuring the \emph{consistency} of pacing annotations.
Indeed, it is easy to write annotations that are incompatible with one another
or that over-constrain a stream by requiring it to produce a value at a time when its dependencies are not guaranteed to be available.
In this paper, we propose a type-based approach to solve this problem.
Concretely, we view pacing annotations as types, and we develop a type system
to check the consistency of pacing annotations.

\paragraph{Contributions}

In summary, we make the following contributions:
\begin{enumerate}
    \item We present \CLola, a core language for stream-based monitoring that models the essential features of \textsc{RTLola}, and in particular, its system of pacing annotations.
    \item We equip \CLola with a formal semantics, interpreting specifications as relations between sets of input and output streams.
    \item We present a type system for \CLola that checks the consistency 
    of pacing annotations.
    \item Using the method of logical relations, we prove that the type system is \textit{sound}, ensuring that well-typed specifications define a total input-output relation.
    This guarantees that well-typed specifications can be compiled into safe executable programs.
\end{enumerate}
% Furthermore, we note that, as we proof the soundness of the type systems against a relational semantics, a specification is not only sound considering a specific evaluation algorithm, but it proofs that there always is a sound way to evaluate the specification.

\paragraph{Structure of the paper}

\Cref{sec:overview} provides an overview of \textsc{RTLola} and its system of pacing annotations. It also illustrates the inconsistency problem and its type-based resolution.
\Cref{sec:clola} formalizes the syntax and the semantics of \CLola, a core fragment of \textsc{RTLola}.
\Cref{sec:pacing_types} presents a first simplified version of the pacing type system and its proof of soundness. This first version of the type systems prohibits output streams from accessing their past values.
We lift this restriction in \Cref{sec:pacing_types_self} before presenting related work in \Cref{sec:related} and concluding the paper in \Cref{sec:conclusion}.

%

%Unlike synchronous languages, streams may not evaluate synchronously.
%This aligns with the nature of cyber-physical system observations, which are often asynchronous.
%For instance, two system sensors may generate measurements at entirely distinct time points.
%Since stream-based monitoring languages impose no constraints on the input stream timing, they provide a declarative method to reason about the safety of such systems.
%
%However, the complexity of stream-based languages, which enables their usage for monitoring CPS, also complicates the development of formally sound analyses that provide static guarantees on runtime behavior.
%In this paper, we present \CLola, a condensed variant of RTLola that balances RTLola’s high-level expressiveness and the feasibility of static analysis steps.
  
%  \subsection{Contributions}

  \section{Overview}\label{sec:overview}
This section provides a high-level overview of \textsc{RTLola} and its pacing annotations using a battery management system as an example.
In \textsc{RTLola}, monitors are specified by declaring a list of named input and output streams.
As an example, consider the following \textsc{RTLola} specification which 
monitors excessive power draw:
  \begin{lstlisting}[label=ex:starting, caption={A specification for monitoring excessive power draw.}]
input  battery_level
output drain   := battery_level.prev(or: battery_level) - battery_level
output warning := drain > 5
  \end{lstlisting}
The input stream \stream{battery\_level} captures the current battery level measurement.
The output stream \stream{drain} computes the charge loss by subtracting the current battery level from the previous one, accessed via \lstinline|battery_level.prev(...)|.
The argument to \lstinline|prev| provides a default value for the first evaluation step.
This example illustrates the core concepts of input and output streams and the use of stream offsets to express temporal properties.

\headline{Handling Multiple Asynchronous Inputs}

In practice, a monitor often requires access to multiple sensors to produce meaningful verdicts about the system's health.
For example, the previous specification can be refined to monitor that the battery does not overheat while it is charging.
For this, an additional input stream \stream{temperature} captures the current temperature of the battery being monitored:
\begin{lstlisting}[label=ex:async_spec, caption={An asynchronous Specification.}]
  input battery_level: Int
  input temperature: Int

  output drain @battery_level@ := battery_level.prev(or: battery_level) - battery_level
  output temp_warning @battery_level & temperature@ := drain < 0 && temperature > 50
\end{lstlisting}

Importantly, \textsc{RTLola} does not make the assumption that streams are synchronized: different inputs might arrive at different rates.
This raises the question of how to interpret an expression like \lstinline{drain < 0 && temperature > 50} in the previous specification.
More precisely, it is not clear what the value of this expression should be if one of the two inputs is not available at the time of evaluation.
For example, the monitor could take the last defined value of the stream or wait for all values to be available.
However, both solutions would negatively impact the precision of the monitor.
Implicitly taking the last defined values would allow to mix
measurements performed at very different time points and whose comparison is, therefore, difficult to interpret.
Similarly, waiting for all values to arrive might result in verdicts being missed.
Instead, \textsc{RTLola} permits more fine-grained handling of the asynchronous arrival of data via a combination of \emph{pacing annotations} and a selection of synchronous and asynchronous access operators.

In the previous specification, note that the outputs \stream{drain} and \stream{temp\_warning} are decorated with the pacing annotations \lstinline{@battery_level} and \lstinline{@batterly_level & temperature} respectively.
These pacing annotations are positive boolean formulas over inputs that define when a stream should be evaluated depending on the availability of inputs.
In the above example, \stream{drain} is evaluated whenever a new battery level reading is received, while \stream{temp\_warning} is evaluated only when both a new battery level and temperature reading arrive simultaneously.

\headline{Synchronous and Asynchronous Stream Accesses}

The \stream{temp\_warning} stream issues a temperature warning based on the current values of the \stream{drain} and \stream{temperature} streams.
Such a direct access to the current value of a stream is called a \emph{synchronous} access.
A synchronous access requires the stream being accessed (i.e., the stream the access refers to) to have a value at the same time as the accessing stream (i.e., the stream that contains the access in its expression).
This requirement ensures that the accessed value is indeed available when needed.
In fact, \lstinline|prev| is, by design, also a synchronous access.
% {\color{red}as only if the accessing and accessed stream produce values at compatible pace, the previous value is unique and fresh}\todo{sentence too long and difficult to parse}.

However, restricting \stream{temp\_warning} to evaluate only when both inputs receive a value simultaneously is problematic because different sensors generate values at arbitrary, non-overlapping time points.
This could prevent \stream{temp\_warning} from ever being evaluated, thus never issuing the intended warning.
\Cref{fig:timing_diagram:sync} illustrates this issue with an example trace.
Black continuous arrows represent synchronous dependencies, while light gray arrows indicate dependencies due to default values.
Observe that the \stream{temp\_warning} stream is only evaluated at time three when both \stream{temperature} and \stream{battery\_level} receive new values.

To address this, \textsc{RTLola} provides an asynchronous \lstinline|hold| access.
A \lstinline|hold| access retrieves the last produced value of the accessed stream, regardless of when it was generated.
If no such value exists, a default value is returned, similar to \lstinline|prev|.
Unlike \lstinline|prev|, the \lstinline|hold| operator may also reference the current value if the accessed stream produces a value at the same time as the accessing stream.

In summary, \textsc{RTLola} provides three operators for accessing stream values:
\begin{itemize}
  \item A direct stream access references the current value of a stream, which must exist at the time of access.
  \item An asynchronous \lstinline|hold| access returns the most recent stream value or a default value if none exists, effectively approximating the accessed stream to ensure a value at every time point.
  \item The \lstinline|prev| access returns the last value of a stream, requiring that the stream also produces a value at the current time step. This captures the intent of referring to a past value of a stream relative to the pacing of the accessing stream. This is especially useful to compute system dynamics like, for example, in \Cref{ex:starting}.
 
%  bridges the expressivity gap between the other two by always returning the second-to-last value of a stream. % Precisely describe what prev does. Mentioning that it is useful to calculate system dynamics.
%  Since \lstinline|prev| is a synchronous access, it requires the accessed stream to produce values simultaneously with the accessing stream.
%  Thus, \lstinline|prev| references a \emph{recent} past value of the stream, unlike \lstinline|hold|, which provides no guarantees on the age of the returned value, or direct access, which cannot refer to past values.
\end{itemize}

Using a \lstinline|hold| access the specification in \Cref{ex:async_spec} can be corrected as follows:
\begin{lstlisting}[label=ex:async_spec_hold, caption={An asynchronous Specification using a Hold Access.}]
  input battery_level: Int
  input temperature: Int

  output drain @battery_level@ := battery_level.prev(or: battery_level) - battery_level
  output temp_warning @battery_level@ := drain < 0 && temperature.hold(or: 0) > 50
\end{lstlisting}
Replacing the synchronous access to \stream{temperature} in \stream{temp\_warning} with a \lstinline|hold| access allows \stream{temp\_warning} to be evaluated whenever a new battery level is measured, reflected by its pacing annotation.
This eliminates the requirement for \stream{temperature} and \stream{battery\_level} to update simultaneously.
\Cref{fig:timing_diagram:hold} illustrates this timing behavior with an example trace.
Continuous black arrows denote synchronous dependencies, dashed arrows indicate asynchronous dependencies and light gray arrows represent default values.
Unlike in \Cref{fig:timing_diagram:sync}, \stream{temp\_warning} is now evaluated at time steps one and five as well.

\begin{figure}[H]
  \begin{subfigure}{0.48\textwidth}
    \includegraphics[width=\linewidth]{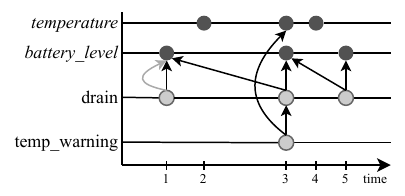}
    \caption{A timing Diagram of the Specification in \Cref{ex:async_spec}}
    \label{fig:timing_diagram:sync}
\end{subfigure}
\hfill
\begin{subfigure}{0.48\textwidth}
    \includegraphics[width=\textwidth]{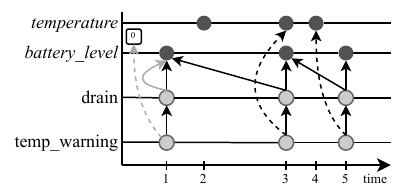}
    \caption{A timing Diagram of the Specification in \Cref{ex:async_spec_hold}}
    \label{fig:timing_diagram:hold}
\end{subfigure}
\caption{A Timing Diagram for the Specification in \Cref{ex:async_spec} and \Cref{ex:async_spec_hold}.}
\label{fig:timing_diagram}
\end{figure}

\headline{Pacing Annotations as a Programming Paradigm}
One could argue that the intent of the \stream{temp\_warning} stream is still not fully captured, as it updates only when the battery level changes, missing potential safety issues at times two and four in \Cref{fig:timing_diagram:hold}.

This issue underscores the need for expressive pacing annotations that go beyond merely tracking dependencies in stream equations.
Instead, the pacing should guide the selection of stream accesses to accurately reflect the intended behavior of a stream.

Applied to the example in \Cref{ex:async_spec_hold}, we first modify the pacing annotation of \stream{temp\_warning} such that the safety condition is checked whenever either the temperature \emph{or} the battery level changes.
We then update its expression to use hold accesses for both streams.
The resulting specification is shown in \Cref{ex:async_spec_hold_hold}.
\begin{figure}[H]
  \centering
    \begin{subfigure}{0.55\textwidth}
\begin{lstlisting}
  input battery_level: Int
  input temperature: Int

  output drain @battery_level@ := 
      battery_level.prev(or: battery_level) - battery_level
  output temp_warning @battery_level | temperature@ := 
      drain.hold(or: 0) < 0 && temperature.hold(or: 0) > 50
\end{lstlisting}
      \caption{A Specification using disjunctive Pacing Types.}
      \label{ex:async_spec_hold_hold}
  \end{subfigure}
  \hfill
  \begin{subfigure}{0.44\textwidth}
\includegraphics[width=\textwidth]{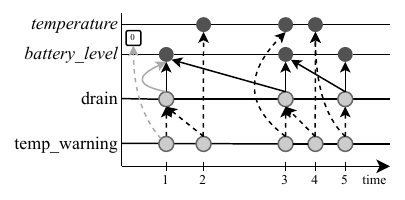}
    \caption{A timing Diagram of the Specification.}
    \label{fig:timing_diagram:hold_hold}
  \end{subfigure}
  \caption{Disjunctive Pacing Types Demonstrated.}
  \label{fig:disjunctive_pacing}
\end{figure}
The timing diagram in \Cref{fig:timing_diagram:hold_hold} shows that, unlike the specification in \Cref{ex:async_spec_hold}, the updated specification evaluates the temperature warning at times two and four, aligning with its intended purpose.
This demonstrates that pacing annotations serve as a mechanism to express temporal intent, enhancing the language’s expressiveness.

\subsection{The Problem}\label{sec:problem}

In practice, monitors generated from specifications must not fail at runtime.
Pacing annotations control data flow within the monitor.
However, not all \textsc{RTLola} programs ensure valid data flow.
Consider the invalid specification in \Cref{fig:invalid:spec}:
\begin{figure}[H]
  \centering
    \begin{subfigure}{0.3\textwidth}
      \begin{lstlisting}
        input a: Int
        input b: Int
        output x @b@ := b
        output y @a@ := x
      \end{lstlisting}
      \caption{An invalid specification}
      \label{fig:invalid:spec}
  \end{subfigure}
  \qquad
  \begin{subfigure}{0.35\textwidth}
      \includegraphics[width=\textwidth]{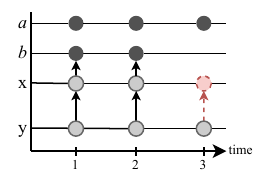}
      \caption{The timing Problem visualized.}
      \label{fig:invalid:timing}
  \end{subfigure}
  \caption{An invalid Example.}
  \label{fig:invalid}
\end{figure}
The pacing annotations require that \stream{y} evaluates whenever \stream{a} receives a new value.
The value of \stream{y} at these times should be the current value of \stream{x}.
Stream \stream{x} must evaluate whenever \stream{b} receives a new value.
If \stream{a} receives a value at runtime but \stream{b} does not, \stream{y} cannot be computed because \stream{x} does not produce a value at this time.
This issue is illustrated in \Cref{fig:invalid:timing}.
Black arrows represent synchronous dependencies between values.
The red dashed arrow and circle at time three indicate the invalid access to the missing value of \stream{x}.
The goal of the type system presented in this paper is to reject such invalid specifications, ensuring that a stream value \emph{can} always be computed when it \emph{must} be computed.

One way to \enquote{repair} the specification in \Cref{fig:invalid:spec} is to replace the synchronous access to \stream{x} in \stream{y} with a \lstinline|hold| access.
However, this change alters the semantics by removing the requirement that \stream{x} must produce a value when \stream{y} evaluates.
With the \lstinline|hold| operator, \stream{y} might always access the same older value of \stream{x} if \stream{x} stops producing values.

\subsection{A Type-driven Solution}

To prevent timing inconsistencies, \textsc{RTLola} uses a type system.
First, we observe that the asynchronous \lstinline|hold| stream access never causes timing inconsistencies, as it inherently provides a default value when the accessed value is absent.
Consequently, the type system focuses on synchronous stream accesses, such as direct or \lstinline|prev| accesses.
More precisely, the type system ensures that a synchronously accessed value is always available when needed at runtime.
To enforce this property, it requires that a stream accessed synchronously evaluates at least at the same time points as the stream containing the access.
This requirement is directly reflected in the example typing rule presented in \Cref{fig:sync_access_rule}.
This typing rule states that accessing
a stream $x$ at pacing $\tau_\mathit{must}$ requires the accessed stream $x$ to have a well-defined value at least at pacing $\tau_\mathit{can}$, where $\tau_\mathit{can}$ is a \emph{finer} annotation than $\tau_\mathit{must}$. This constraint is reflected by the
premise $\tau_\mathit{must} \models \tau_\mathit{can}$ which intuitively ensures that the time points captured by $\tau_\mathit{must}$ are included in those captured by $\tau_\mathit{can}$.

\begin{figure}[ht!]
  \begin{mathpar}
    \inferrule{x \in \mathbb{V}_\texttt{out}\\ (x : \tau_\mathit{can}) \in \Gamma\\ \tau_\mathit{must} \models \tau_\mathit{can}}{\Gamma \vdash x : \tau_\mathit{must}}
  \end{mathpar}
  \caption{Example of typing rule for a direct output stream access}
  \Description{Example of typing rule}
  \label{fig:sync_access_rule}
\end{figure}

In \Cref{sec:clola}, we present the complete type system and prove that it is \emph{sound} in the sense that well-typed specifications are guaranteed to be free of timing inconsistencies.

  \section{\CLola}\label{sec:clola}

\newcommand{\Vin}{\mathbb{V}_\mathit{in}}
\newcommand{\Vout}{\mathbb{V}_\mathit{out}}
\newcommand{\syn}[1]{{\color{brown}\textbf{\texttt{#1}}}}
\newcommand{\stradd}[2]{#1 \ \syn{+} \ #2}
\newcommand{\streq}[2]{#1 \ \syn{==} \ #2}
\newcommand{\strand}[2]{#1 \ \syn{\&\&} \ #2}
\newcommand{\strnot}[1]{\syn{not} \ #1}
\newcommand{\strlt}[2]{#1 \ \syn{<} \ #2}
\newcommand{\strlast}[2]{#1 \syn{.\texttt{prev}(}#2 \syn{)}}
\newcommand{\strhold}[2]{#1 \syn{.\texttt{hold}(}#2 \syn{)}}
\newcommand{\strdef}[3]{#1 \ \syn{@} \ #2 \ \syn{:=} \ #3}
\newcommand{\strwhen}[4]{#1 \ \syn{@} \ #2 \ \syn{when} \ #3 \ \syn{:=} \ #4}
\newcommand{\rhoIn}{{\rho_\textit{in}}}
\newcommand{\rhoOut}{{\rho_\textit{out}}}
\newcommand{\last}[2]{\mathit{Last}(#1, #2)}
\newcommand{\sync}[2]{\mathit{Sync}(#1, #2)}
\newcommand{\prev}[3]{\mathit{Prev}(#1, #2, #3)}
\newcommand{\hold}[3]{\mathit{Hold}(#1, #2, #3)}
\newcommand{\den}[1]{\llbracket #1 \rrbracket}

To formalize pacing types and their properties, we present \CLola, a core language of stream equations.
In essence, \CLola programs represent a set of
constraints over input and output \emph{stream variables}
drawn from (disjoint) sets $\Vin$ and $\Vout$ respectively.
For a fixed set of (named) input streams, these
constraints describe which set of (named) output
streams are valid reactions to the inputs.
Because of this declarative view, we refer to \CLola programs as \emph{specifications}.
For clarity, we only formalize constraints over integer-valued streams, but \CLola can be extended to arbitrary data types in a straightforward manner.

\subsection{Syntax of \CLola}

Constraints on stream variables are expressed via a list of equations of the form
$\strdef{x}{\tau}{e}$. 
In these equations, $x \in \Vout$ is the name of an output stream being constrained, and $e$ is a \emph{stream expression} over input and output variables. Finally, $\tau$ is a positive boolean formula over input variables describing when the output stream $x$ must be evaluated. We call $\tau$ a \emph{pacing annotation}.

Stream expressions can be constants values $v \in \mathbb{Z}$, stream variables $x \in \Vin \uplus \Vout$, arithmetic functions applied to sub-expressions, references to previous values $\strlast{x}{e}$, or \emph{hold} accesses $\strhold{x}{e}$.
In $\strlast{x}{e}$ and $\strhold{x}{e}$,
$x \in \Vin \uplus \Vout$ is the name of an arbitrary input or output stream, and $e$ is a \emph{default expression} representing what value the expression should take if there is no previous value for $x$ (in the case of \syn{prev}) or if there is no value to hold (in the case of \syn{hold}). The exact syntax of specifications, stream expressions, and pacing annotations is defined as:
\vspace{1em}
  \begin{center}
    \begin{tabular}{ll}
      \textit{Variables}          & $x \in \mathbb{V}_\texttt{in} \uplus \mathbb{V}_\texttt{out}$\\
      \textit{Values}             & $v \in \mathbb{Z}$\\
      \textit{Stream Expressions} & $e \ ::= \ v \mid x \mid \strlast{x}{e} \mid \strhold{x}{e} \mid \stradd{e_1}{e_2} \mid \ldots$\\
      \textit{Pacing Annotations} & $\tau ::= x \mid \syn{$\top$} \mid \tau_1 \ \syn{$\wedge$} \ \tau_2 \mid \tau_1 \ \syn{$\vee$} \ \tau_2 $\\
      \textit{Equations}          & $eq ::= \strdef{x}{\tau}{e}$\\
      \textit{Specifications}     & $S ::= \epsilon \mid eq \cdot S$
    \end{tabular}
  \end{center}

\subsection{Semantics of \CLola}

\headline{Streams}

We equip \CLola specifications with a denotational semantics assigning to each spec $S$ a relation $\den{S}$ between input and output streams. Importantly, to account for the asynchronous arrival of inputs (and the asynchronous production of outputs), streams are defined as sequences of \emph{optional} values: at each point in time, a stream can either have a value $v \in \mathbb{Z}$, or be undefined denoted by $\bot$. Formally, the set of streams is defined as follows: \[
  \mathit{Stream} \triangleq (\mathbb{Z} \uplus \{ \bot \})^\mathbb{N}
\]

\headline{Maps}

In order to formally define the denotational semantics of \CLola, we introduce a suitable notion of \emph{stream maps} to represent sets of named streams. Given a set of variables $X \subseteq \Vin \uplus \Vout$, the set of stream maps over $X$ is \[
  \mathit{Smap}(X) \triangleq \mathit{Stream}^X
\]

Given a map $\rho \in \mathit{Smap}(X)$, we note $\mathit{dom}(\rho) \triangleq X$ the domain of $\rho$. We note $\emptyset$ the unique map in $\mathit{Smap}(\emptyset)$.
Given $x \in \Vout$ and $w \in \mathit{Stream}$ we note $(x \mapsto w) \in \mathit{Smap}(\{ x \})$ the singleton map associating $w$ to $x$.
Finally, given two maps $\rho_X \in \mathit{Smap}(X)$ an $\rho_Y \in \mathit{Smap}(Y)$, we define the joined map $\rho_X \cdot \rho_Y$ as follows:
$\rho_X \cdot \rho_Y \in \mathit{Smap}(X \cup Y)$ as follows: \[
  \rho_1 \cdot \rho_2 \triangleq \lambda x.\begin{cases}\begin{tabular}{ll}
    $\rho_1(x)$ &if $x \in \mathit{dom}(\rho_1)$\\
    $\rho_2(x)$ &if $x \notin \mathit{dom}(\rho_1)$ and $x \in \mathit{dom}(\rho_2)$
  \end{tabular}\end{cases}
\]

Importantly, we note that $\cdot$ is \emph{not} a commutative in general but only when
when $X$ and $Y$ are disjoints, as stated by \Cref{lem:disjoint-maps}.

\begin{lemma}\label[plemma]{lem:disjoint-maps} Let $\rho \in \mathit{Smap}(X)$ and $\rho' \in \mathit{Smap}(X')$. \[
  \mathit{dom}(\rho) \cap \mathit{dom}(\rho') = \emptyset \implies \forall x \in \mathit{dom}(\rho). \ (\rho \cdot \rho')(x) = (\rho \cdot \rho')(x) = \rho(x)
\]
\end{lemma}

In the rest of the paper, $\rhoIn$ and $\rhoOut$ are variables ranging over $\mathit{Smap}(\Vin)$ and $\mathit{Smap}(\Vout)$ respectively. We also use
the notation $\rhoOut$ for variables ranging over $\mathit{Smap}(X)$ where $X$ is only a subset of $\Vout$.

 \headline{Semantics of Stream Expressions}

We defining the semantics of stream expressions. Given a stream map $\rho \in \mathit{Smap}(\Vin \uplus \Vout)$, and a time point $n \in \mathbb{N}$, an expression $e$ will be assigned a value $\den{e}^n_\rho \in \mathbb{Z} \uplus \{ \bad \}$ where $\bad$ indicates that the evaluation of the expression \emph{failed}.
To assign a meaning to direct stream accesses and to the operators $\strlast{-}{-}$ and $\strhold{-}{-}$, we first introduce the following functions $\mathit{Sync}$, $\mathit{Last}$, $\mathit{Prev}$ and $\mathit{Hold}$.
In all definitions below, $w \in \mathit{Stream}$ is a stream of optional values, $n \in \mathbb{N}$ is a time point and $v \in \mathbb{Z} \cup \{ \bad \}$.

\begin{definition}[Stream Operators]
  \begin{align*}
      \sync{w}{n} &\triangleq \begin{cases}
        \begin{tabular}{ll}
          $w(n)$ &if $\ w(n) \neq \bot$\\
          $\bad$ &if $\ w(n) = \bot$
        \end{tabular}
      \end{cases}\\
    \last{w}{n} &\triangleq \begin{cases}
        \begin{tabular}{ll}
          $w(n)$ &if $\ w(n) \neq \bot$\\
          $\last{w}{n-1}$ &if $\ w(n) = \bot \land n > 0$\\
          $\bad$ &if $\ w(n) = \bot \land n = 0$
        \end{tabular}
      \end{cases}\\
      \hold{w}{n}{v} &\triangleq \begin{cases}
        \begin{tabular}{ll}
          $\last{w}{n}$ &\quad \ \ if $\last{w}{n} \neq \bad$\\
          $v$ &\quad \ \ if $\last{w}{n} = \bad$
        \end{tabular}
      \end{cases}\\
      \prev{w}{n}{v} &\triangleq \begin{cases}
        \begin{tabular}{ll}
          $\last{w}{n-1}$ &if $w(n) \neq \bot \land n > 0 \land \last{w}{n-1} \neq \bad$\\
          $v$ &if $w(n) \neq \bot \land (n = 0 \vee n > 0 \land \last{w}{n-1} = \bad)$\\
          $\bad$ &if $w(n) = \bot$
        \end{tabular}
      \end{cases}
  \end{align*}
\end{definition}

The intuition behind these functions can be described as follows: \begin{itemize}[leftmargin=*]
  \setlength\itemsep{.8em}
  \item $\mathit{Sync}(w, n)$ is the integer value of the stream $w$ at time $n$, or $\bad$ if $w$ is undefined at time $n$.
  \item $\last{w}{n}$ is the last defined value of $w$ before the time point $n$ (included).
  If no such value exists (i.e. $w$ never had a defined value before time point $n$ included),
  $\last{w}{n}$ fails and returns $\bad$.
  \item $\hold{w}{n}{v}$ is also the last defined value of $w$ from $n$ (included).
  However, contrary to $\last{w}{n}$, when such a value does not exists (i.e. when $\last{w}{n} = \bad$), $\hold{w}{n}{v}$ produces the default value $v$ instead.
  Said otherwise, $\lambda n. \hold{w}{n}{v}$ can be seen as an approximation of the stream $w$ that is guaranteed to have a defined value at all times.
  \item $\prev{w}{n}{v}$ is slightly more involved.
  To evaluate $\strlast{x}{e}$, we first inspect the current value of stream $x$. If $x$ is defined at the current time point,
  $\strlast{x}{e}$ returns the previous defined value of $x$. If no such value exists (i.e.
  the current time point is the first time point where $x$ is defined) a default value (defined by $e$) is returned instead.
  This intuition is exactly captured by the function $\prev{w}{n}{v}$: it first checks whether
  $w(n)$ is defined. If $w(n) = \bot$, $\mathit{Prev}$ fails and returns $\bad$.
  Otherwise, if $w(n)$ is defined, $\mathit{Prev}$ returns the last defined value
  of $w$ (excluding $w(n)$). If $n$ is the first time point where $w$ is defined, $\mathit{Prev}$
  returns the default value $v$.
\end{itemize}

Using the functions defined above, we define the denotation of stream expressions. 
Constants and arithmetic operations are evaluated in an obvious way. Stream accesses are evaluated using $\mathit{Sync}$, $\mathit{Prev}$ and $\mathit{Hold}$.

\begin{definition}[Denotation of stream expressions]
  \begin{align*}
    \den{c}^n_\rho &\triangleq c\\
    \den{\stradd{e_1}{e_2}}^n_\rho &\triangleq \den{e_1}^n_\rho +_{\bad} \den{e_2}^n_\rho\\
    \den{x}^n_\rho &\triangleq \sync{\rho(x)}{n}\\
    \den{\strlast{x}{e}}^n_\rho &\triangleq \prev{\rho(x)}{n}{\den{e}_\rho^n}\\
    \den{\strhold{x}{e}}^n_\rho &\triangleq \hold{\rho(x)}{n}{\den{e}_\rho^n}
  \end{align*}
  Where the operator $+_{\bad}$ is defined as follows:\[
  v_1 +_{\bad} v_2 \triangleq \begin{cases}
    \begin{tabular}{ll}
      $v_1 + v_2$ &if $v_1 \neq \bad \land v_2 \neq \bad$\\
      $\bad$ &otherwise
    \end{tabular}
  \end{cases}
\]
\end{definition}

\headline{Semantics of Pacing Annotations}
Pacing annotations define the discrete time points at which a stream is evaluated.
Hence, they are interpreted as a set of natural numbers that define these time points.
Syntactically, they are positive boolean formulas over input streams.
A literal in the boolean formula, i.e. an input stream variable, corresponds to the set of time points at which this stream has a defined value different from $\bot$.
The boolean operators then follow naturally from the set intersection and union.
Formally, given $\rhoIn \in \mathit{Smap}(\Vin)$, a pacing annotation $\tau$ denotes the set of time points $\den{\tau}_\rhoIn \subseteq \mathbb{N}$ defined below.

\begin{definition}[Denotation of pacing annotations]
  \begin{align*}
    \den{x}_\rhoIn &\triangleq \{ n \mid \rhoIn(x)(n) \ne \bot \}\\
    \den{\syn{$\top$}}^n_\rhoIn &\triangleq \mathbb{N}\\
    \den{\tau_1 \ \syn{$\vee$} \ \tau_2}_\rhoIn &\triangleq \den{\tau_1}_\rhoIn \cup \den{\tau_2}_\rhoIn\\
    \den{\tau_1 \ \syn{$\wedge$} \ \tau_2}_\rhoIn &\triangleq \den{\tau_1}_\rhoIn \cap \den{\tau_2}_\rhoIn
  \end{align*}
\end{definition}

\headline{Semantics of Specifications}

Having defined the semantics of expressions and pacing annotations, we can define the semantics of equations and specifications. An equation $\strdef{x}{\tau}{e}$ denotes a set of maps $\rho \in \mathit{Smap}(\Vin \cup \Vout)$. A specification (i.e., a list of equations) is interpreted as the intersection of the solutions of its equations.

\begin{definition}[Denotation of Equations]
  \begin{align*}
    \den{\strdef{x}{\tau}{e}} &\triangleq \{ \ \rho = \rhoIn \cdot \rhoOut \mid \forall n \in \den{\tau}_\rhoIn. \rhoOut(x)(n) = \den{e}^n_\rho \ \}\\
    \den{ \epsilon } &\triangleq \mathit{Smap}(\mathbb{V}_\texttt{in}) \cdot \mathit{Smap}(\mathbb{V}_\texttt{out})\\
    \den{ \mathit{eq} \cdot S } &\triangleq \den{\mathit{eq}} \cap \den{S}
  \end{align*}
\end{definition}

The denotation of a single equation $\strdef{x}{\tau}{e}$ is defined to be the set of all maps consistent with the constraint $x = e$.
However, this requirement is refined by only requiring the constraint $x = e$ to hold at the time points described by the pacing annotation $\tau$.
We note that since $\rhoOut(x)(n) \in \mathbb{Z} \cup \{ \bot \}$ and $\den{e}^n_\rho \in \mathbb{Z} \cup \{ \bad \}$, both $x$ and $e$ must have
a well-defined value in $\mathbb{Z}$ at all points in time described by the pacing annotation $\tau$.
% Thex purpose of our pacing type system is precisely to enforce that this is indeed the case, for any possible inputs.

% and with the pacing annotation $\tau$. Formally, it is precisely the set of maps $\rho \rhoIn \cdot \rhoOut$ such that at all time points described by the pacing annotation $\tau$, the equation $x $

  % \todo[inline]{it's all about the side condition on e not bot in the case of equation}
  % \todo{THE BIG GOAL IS:
  %   it should not be possible to have on the right of x @ t = e
  %   an expresssion e that is not guaranteed to have a value at t
  %   (this is what the side condition is forcing)
  % }

  % {\color{red}
  %   VERY IMPORTANT $\tau$ (i;e., pacing annot) is telling you when you HAVE to produce a value.
  %   The job of the pacing checker is to compute when you CAN produce a value an make
  %   sure that HAVE-TO <= CAN.

  %   Now additionally, c (filter) further limit by reducing the have-to set!
  %   It says when you HAVE-TO produce a value, if condition (c) does not hold you can instead produce nothing.
  %   Note that importantly, c contributes to the CAN set (because c can only be evaluated if the values it depends on are there)
  % }

  \section{Pacing Types Without Self References}
\label{sec:pacing_types}

In the previous section, we defined the semantics
of specifications by assigning to each lists of equations a
set of compatible stream maps.
An important question to ask is whether there is at least one solution to a list of equations.
As discussed in the overview, this is not always the case: it is easy to write
specifications that do not have any model. Worse, it is possible to write
specifications that only have models for specific combinations of inputs.
The goal of this paper is to provide a type system that enforces the existence of a solution for \emph{any} combination of inputs.

\begin{definition}[Safety]
  A specification $S$ is called \emph{safe} if, for any set of
  named input streams $\rhoIn$, there exists at
  least one compatible set of named output streams $\rhoOut$. Formally \[
    \mathit{Safe}(S) \triangleq \forall \rhoIn \in \mathit{Smap}(\Vin). \exists \rhoOut  \in \mathit{Smap}(\Vout). \ \rhoIn \cdot \rhoOut \in \den{S}
  \]
\end{definition}

In the following, we will define a type system for specifications that ensures safety.
Our type system operates on two levels: first, we define a typing judgment $\Gamma \vdash e : \tau$ for expressions, where $\Gamma$ is a \emph{typing context}, $e$ is a stream expression, and $\tau$ is a pacing annotation.
We also define a typing judgment $\Gamma \vdash S$ to type lists of equations $S$. The typing rules for lists of equations will exploit the typing judgment for expressions.
Typing contexts $\Gamma$ are partial maps from output variables $x \in \Vout$ to annotations $\tau$. We adopt the following usual notation convention $\Gamma, x : \tau \triangleq \Gamma[x \mapsto \tau]$ and $x : \tau \in \Gamma \iff \Gamma(x) = \tau$. 

\subsection{Typing Rules}\label{sec:rules-v1}

As discussed before, the typing rules for \CLola presented operate on two levels: specifications (i.e., lists of equations) and expressions.
On the specification level, the typing context $\Gamma$ \emph{collects} knowledge about the pacing of output streams defined by equations.
On the expression level, the typing rules enforce that a sub-expression can be evaluated, under the assumption that each stream $x : \tau \in \Gamma$ is $\tau$-paced.

\headline{Typing Expressions}

The pacing types rules of \CLola for expressions are presented in \Cref{fig:typing-v1}.
As expected, the interesting type constraints arise from stream accesses.
Each stream access kind comes associated with two typing rules.
Which rule is applicable depends on whether the accessed stream is an input or output stream, distinguished by the \emph{In} and \emph{Out} suffixes of the rule names.

Intuitively, the synchronous accesses, i.e., rules \textsc{DirectOut}, \textsc{DirectIn}, \textsc{PrevOut}, and \textsc{PrevIn}, require that the time points at which the expression \emph{must} be evaluated are subsumed by the time points at which the accessed stream \emph{can} be evaluated.
This is made explicit by the requirement $\tau_\mathit{must} \models \tau_\mathit{can}$ which denotes that type $\tau_\mathit{must}$ refines $\tau_\mathit{can}$.
Formally, this relation is defined as follows:
\[
\tau_\mathit{must} \models \tau_\mathit{can} \triangleq \forall \rhoIn. \den{\tau_\mathit{must}}_\rhoIn \subseteq \den{\tau_\mathit{can}}_\rhoIn
\]
Here, $\tau_\mathit{must}$ is the expected type annotated to the stream equation the expression is associated with.
Whereas $\tau_\mathit{can}$ is either expected to be in the typing context for output streams, i.e. the accessed output stream was correctly typed before, or it is the stream name itself in case of input streams, as pacing types are positive boolean formulas over input stream names.
Additionally, the previous access, i.e. the rules \textsc{PrevOut} and \textsc{PrevIn} require that the default expression can be evaluated at the same times as the access, expressed through the $\Gamma \vdash e : \tau_\mathit{must}$ premise.

Similarly, a hold access also requires that the default expression can be evaluated at the same times as the access itself, yet, it does not require the accessed stream to be of compatible pacing.
On the contrary, hold accesses only require the accessed stream to exist expressed through the $x \in \mathit{dom}(\Gamma)$ and $x \in \Vin$ premises for output and input streams, respectively.
This coincides with the intuition, that a hold access does not require the accessed stream value to exist, as the default value can always be chosen if it does not.
\begin{figure}[ht!]
\begin{mathpar}
  \inferrule[Const]{v \in \mathbb{Z}}{\Gamma \vdash v : \tau}\and
  \inferrule[BinOp]{\Gamma \vdash e_1 : \tau_\mathit{must} \\ \Gamma \vdash e_2 : \tau_\mathit{must} }{\Gamma \vdash \stradd{e_1}{e_2} : \tau_\mathit{must}}\\
  \inferrule[DirectOut]{(x : \tau_\mathit{can}) \in \Gamma\\ \tau_\mathit{must} \models \tau_\mathit{can}}{\Gamma \vdash x : \tau_\mathit{must}}\and
  \inferrule[DirectIn]{x \in \mathbb{V}_\texttt{in} \\ \tau_\mathit{must} \models x }{\Gamma \vdash x : \tau_\mathit{must}}\\
  \inferrule[PrevOut]{(x : \tau_\mathit{can}) \in \Gamma \\ \Gamma \vdash e : \tau_\mathit{must} \\ \tau_\mathit{must} \models \tau_\mathit{can}\\ }{\Gamma \vdash \strlast{x}{e} : \tau_\mathit{must}}\and
  \inferrule[PrevIn]{x \in \Vin \\ \Gamma \vdash e : \tau_\mathit{must} \\ \tau_\mathit{must} \models x\\ }{\Gamma \vdash \strlast{x}{e} : \tau_\mathit{must}}\\
  \inferrule[HoldOut]{x \in \mathit{dom}(\Gamma)\\ \Gamma \vdash e : \tau_\mathit{must}}{\Gamma \vdash \strhold{x}{e} : \tau_\mathit{must}}\and
  \inferrule[HoldIn]{x \in \Vin \\ \Gamma \vdash e : \tau_\mathit{must} }{\Gamma \vdash \strhold{x}{e} : \tau_\mathit{must}}
\end{mathpar}
  \caption{The \CLola pacing type rules for stream expressions.}
  \Description{typing rules}
  \label{fig:typing-v1}
\end{figure}

\headline{Typing Equations}
The typing rules for equations are shown in \Cref{fig:pacing_types_v1_eqs}.
They process the stream equations in the specification $\phi$ in order of their appearance, extending the typing context for each equation processed.
Trivially, the \textsc{Empty} rules states that any typing context is valid for an empty specification.
The \textsc{Eq} rules processes a single stream equation and recurses for the remaining specification while binding the current stream variable ($x$) to the annotated pacing type $\tau$.
For an equation to be valid, the stream variable should not be bound in $\Gamma$, i.e., $x \notin \mathit{dom}(\Gamma)$, which asserts that there is only a single stream equation for the output stream in the specification.
More importantly, it requires that the expression of the stream equation $e$ can be evaluated with the annotated pacing type $\tau$, stated as the premise $\Gamma \vdash e : \tau$.
\begin{figure}
	\begin{mathpar}
  \inferrule[Empty]{ }{\Gamma \vdash \epsilon} \and 
  \inferrule[Equation]{ x \notin \mathit{dom}(\Gamma)\\\Gamma \vdash e : \tau\\ \Gamma, x : \tau \vdash \varphi }{\Gamma \vdash (\strdef{x}{\tau}{e}) \cdot \varphi }
\end{mathpar}
\caption{The \CLola pacing type rules for stream equations.}
\label{fig:pacing_types_v1_eqs}
\end{figure}

\subsection{A Logical Relation for Pacing Types}

Our goal is to prove that a specification that is well-typed according to the rules presented in the previous section is necessarily safe.
We prove this by employing the method of \emph{logical relations}.
We first construct a semantic interpretation of pacing types, from which we derive a \emph{semantic typing judgment} for specifications $\Gamma \vDash S$.
Importantly, the interpretation of pacing types is defined directly in terms of the semantics of specifications, and such that it only contains safe specifications by definition.
To prove the soundness of our type system, it then suffices to show that syntactically well-typed specifications are also semantically well-typed (i.e., $\Gamma \vdash S \implies \Gamma \vDash S$).

\headline{Semantics of Expressions Under Partial Maps}

\newcommand{\pden}[1]{\widehat{\den{#1}}}

In order to define our semantic interpretation of pacing types, we first define an alternative semantics for stream expressions under \emph{partial stream maps} (i.e., maps that do not necessarily associate a stream to every output stream variable). Given a subset $X \subseteq \Vout$ and $\rho \in \mathit{Smap}(\Vin \cup X)$, we define a new interpretation of expressions $\pden{e}_\rho$ which explicitly returns $\bad$ when $e$ cannot be evaluated because a specific output stream is not available in $\rho$.

\begin{definition}[Denotation of stream expressions under partial maps]
  \label{def:partial-sem}
  \begin{align*}
    \pden{c}^n_\rho &\triangleq c\\
    \pden{\stradd{e_1}{e_2}}^n_\rho &\triangleq \pden{e_1}^n_\rho +_{\bad} \pden{e_2}^n_\rho\\
    \pden{x}^n_\rho &\triangleq \begin{cases}
    	\begin{tabular}{ll}
    		$\sync{\rho(x)}{n}$ &if $x \in \mathit{dom}(\rho)$\\
    		$\bad$ &if $x \not\in \mathit{dom}(\rho)$\\
    	\end{tabular}
    \end{cases}\\
    \pden{\strlast{x}{e}}^n_\rho &\triangleq \begin{cases}
    	\begin{tabular}{ll}
    		$\prev{\rho(x)}{n}{\pden{e}_\rho^n}$ &if $x \in \mathit{dom}(\rho)$\\
    		$\bad$ &if $x \not\in \mathit{dom}(\rho)$\\
    	\end{tabular}
    \end{cases}\\
    \pden{\strhold{x}{e}}^n_\rho &\triangleq \begin{cases}
    	\begin{tabular}{ll}
    		$\hold{\rho(x)}{n}{\pden{e}_\rho^n}$ &if $x \in \mathit{dom}(\rho)$\\
    		$\bad$ &if $x \not\in \mathit{dom}(\rho)$\\
    	\end{tabular}
    \end{cases}
  \end{align*}
\end{definition}

We note that for any expression $e$ and any partial map $\rho$ that contains sufficiently many streams for $e$ to be safely evaluated, then $\den{e}$ and $\pden{e}$ coincide in the sense of the following lemma.

\begin{lemma}
  \label[plemma]{lem:pden-den}
  Let $X \subseteq \Vout$, $\rho \in \mathit{Smap}(\Vin \cup X)$ and $\rho' \in \mathit{Smap}(\Vout \setminus X)$. Then \[
    \pden{e}_\rho \ne \bad \implies \pden{e}_\rho = \den{e}_{\rho \cdot \rho'}
  \]
\end{lemma}
\begin{proof}
  By induction on $e$.
\end{proof}

\headline{Interpretation of Pacing Types}

\newcommand{\Lsafe}[1]{\mathcal{S}\llbracket #1 \rrbracket}
  \newcommand{\Lctx}[1]{\mathcal{G}\llbracket #1 \rrbracket}
  \newcommand{\Lstr}[1]{\mathcal{W}\llbracket #1 \rrbracket}
  \newcommand{\Lequ}[1]{\mathcal{E}\llbracket #1 \rrbracket}

Using definition \ref{def:partial-sem}, we construct our semantic interpretation of pacing types.
We define four unary relations $\Lstr{-}$, $\Lctx{-}$, $\Lequ{-}$, and $\Lsafe{-}$. $\Lstr{-}$ associates pacing annotations with streams, $\Lctx{-}$
associate typing contexts with output stream maps, $\Lequ{-}$ associates typing contexts with expressions, and $\Lsafe{-}$ associates typing contexts with specifications.

\begin{definition}[Interpretation of Pacing Types] \ \begin{align*}
    \Lstr{\tau}_{\rho_\texttt{in}} &\triangleq \{ \ w \in \mathit{Stream} \mid \forall n \in \den{\tau}^n_{\rho_\texttt{in}}. w(n) \ne \bot \ \}\\
    \\
    \Lctx{\Gamma}_{\rhoIn} &\triangleq \{ \ \rhoOut \in \mathit{Smap}(\mathit{dom}(\Gamma)) \mid \forall (x : \tau) \in \Gamma. \ \rhoOut(x) \in \Lstr{\tau}_{\rho_\texttt{in}} \ \}\\
    \\
    \Lequ{\Gamma \mid \tau}_{\rho_\texttt{in}} &\triangleq \{ \ e \mid \forall \rho_\texttt{out} \in \Lctx{\Gamma}_\rhoIn. \forall n \in \den{\tau}_\rhoIn. \pden{e}^n_{\rhoIn \cdot \rhoOut} \ne \bad \ \}\\
    \\
    \Lsafe{\Gamma}_{\rho_\texttt{in}} &\triangleq \{ \ S \mid \forall \rho^1_\texttt{out} \in \Lctx{\Gamma}_{\rho_\texttt{in}}. \exists \rho^2 \in \mathit{Smap}(\Vout \setminus \mathit{dom}(\Gamma)) \wedge \rhoIn \cdot \rho^1_\texttt{out} \cdot \rho^2_\texttt{out} \in \den{S} \ \}
  \end{align*}
\end{definition}

The stream relation $\Lstr{\tau}_\rhoIn$ is straightforward: it contains all streams $w$
that are defined at least at all time points in $\den{\tau}_\rhoIn$.
The context relation $\Lctx{\Gamma}_\rhoIn$ extends $\Lstr{-}$ in a natural way:
it contains all partial maps $\rhoOut \in \mathit{Smap}(\mathit{dom}(\Gamma))$ such that
$\rhoOut(x)$ is in the interpretation of $\tau$ for all bindings $x : \tau$ in $\Gamma$.
The expression relation $\Lequ{\Gamma \mid \tau}_\rhoIn$ captures all expressions $e$ that can be safely evaluated at pacing $\tau$ under any output streams compatible with $\Gamma$.
In other words, for any $\rhoOut$ in the interpretation of $\Gamma$, and for any time point denoted by $\tau$, $e$ should safely evaluate to a value in $\mathbb{Z}$.
Note that since $\rhoOut$ is not necessarily a total map in the definition of $\Lequ{-}$, we use $\pden{-}$ instead of $\den{-}$ to evaluate $e$.
Finally, the specification relation $\Lsafe{\Gamma}_\rhoIn$ contains all specifications $S$ such that any partial map in $\Lctx{\Gamma}_\rhoIn$ can be extended into a total map that is a solution to all equations in $S$.

\headline{Semantic Typing}

Based on the interpretation(s) of pacing types, we define the following semantic typing judgments:
  \begin{align*}
    \Gamma \models S &\triangleq \forall \rhoIn.\ S \in \Lsafe{\Gamma}_\rhoIn\\
    \Gamma \models \rho &\triangleq \forall \rhoIn.\ \rho \in \Lctx{\Gamma}_\rhoIn\\
    \Gamma \models e: \tau &\triangleq \forall \rhoIn.\ e \in \Lequ{\Gamma \mid \tau}_\rhoIn
  \end{align*}

Importantly, we note that specifications that are semantically well-typed are safe by definition.

\begin{theorem}[Safety]
  $\emptyset \vDash S \implies \mathit{Safe}(S)$
  \label{thm:safety}
\end{theorem}
\begin{proof}
  Suppose $\emptyset \vDash S$ and let $\rhoIn \in \mathit{Smap}(\Vin)$. We have to show the existence of some $\rhoOut \in \mathit{Smap}(\Vout)$ such that $\rhoIn \cdot \rhoOut \in \den{S}$.
  Since $\emptyset \vDash S$, in particular $S \in \Lsafe{\emptyset}_\rhoIn$. By definition it follows that \[
    \forall \rho^1_\mathit{out} \in \Lctx{\emptyset}_\rhoIn. \exists \rho^2_\mathit{out} \in \mathit{Smap}(\Vout). \rhoIn \cdot \rho^1_\mathit{out} \cdot \rho^2_\mathit{out} \in \den{S}
  \]
  By picking $\rho^1_\mathit{out} = \emptyset$ we obtain some $\rho^2_\mathit{out}$ such that $\rhoIn \cdot \rho^1_\mathit{out} \cdot \rho^2_\mathit{out} = \rhoIn \cdot \emptyset \cdot \rho^2_\mathit{out} = \rhoIn \cdot \rho^2_\mathit{out} \in \den{S}$.
  therefore, it suffices to pick $\rhoOut \triangleq \rho^2_\mathit{out}$.
\end{proof}

\subsection{Soundness}

\headline{Soundness for Expressions}

We first focus on the soundness of the type system for expressions.
The theorem is stated as follows: \begin{theorem}[Soundness of Typing for Expressions]
  $\Gamma \vdash e : \tau \implies \Gamma \models e : \tau$
\end{theorem}

To prove this statement, we show that all typing rules of \Cref{fig:typing-v1} are compatible with semantic typing. In other words, we show that these typing rules can also be derived for the semantic typing judgment $\Gamma \vDash e : \tau$.

\begin{lemma}[Compatibility of Typing Rules for Expressions]
  \label[plemma]{lem:expression-compat}
  The following implications hold:
  \begin{enumerate}
  \item $v \in \mathbb{Z} \implies \Gamma \models v: \tau$
    \item $x \in \Vin \implies \tau \models x \implies \Gamma \vDash x : \tau$
    \item $(x : \tau_\mathit{can}) \in \Gamma \implies \tau_\mathit{must} \models \tau_\mathit{can} \implies \Gamma \vDash x : \tau_\mathit{must}$
    \item $\Gamma \vDash (e_1 : \tau) \implies \Gamma \vDash (e_2 : \tau) \implies \Gamma \vDash \stradd{e_1}{e_2}: \tau$
    \item $(x : \tau_\mathit{can}) \in \Gamma \implies \Gamma \models (e : \tau_\mathit{must}) \implies \tau_\mathit{must} \models \tau_\mathit{can} \implies \Gamma \models \strlast{x}{e}: \tau_\mathit{must} $
    \item $x \in \Vin \implies \Gamma \models (e : \tau_\mathit{must}) \implies \tau_\mathit{must} \models x \implies \Gamma \models \strlast{x}{e}: \tau_\mathit{must} $
    \item $x \in \mathit{dom}(\Gamma) \implies \Gamma \models (e: \tau) \implies \Gamma \models \strhold{x}{e}: \tau$
    \item $x \in \Vin \implies \Gamma \models (e: \tau) \implies \Gamma \models \strhold{x}{e}: \tau$
  \end{enumerate}
\end{lemma}
\begin{proof}
We focus on three difficult cases, the others are either trivial or similar.

\bigskip\noindent\textbf{(3) Direct access to an output:} Suppose $(x : \tau_\mathit{can}) \in \Gamma$ and $\tau_\mathit{must} \models \tau_\mathit{can}$. We have to show that $\Gamma \models x : \tau_\mathit{must}$. 
By definition of $\Gamma \vDash e$, it suffices to fix an arbitrary $\rhoIn \in \mathit{Smap}(\Vin)$, an arbitrary $\rhoOut \in \Lctx{\Gamma}_\rhoIn$, and it remains to show that \[
  \forall n \in \den{\tau_\mathit{must}}_\rhoIn. \pden{x}^n_{\rhoIn \cdot \rhoOut} \ne \bad
\]

Since $(x : \tau_\mathit{can}) \in \mathit{dom}(\Gamma)$ and $\rhoOut \in \Lctx{\Gamma}_\rhoIn$, we know that $\rhoOut(x) \in \Lstr{\tau_\mathit{can}}_\rhoIn$ and therefore, $\forall n \in \den{\tau_\mathit{can}}. \rhoOut(x)(n) \ne \bot$.
Additionally, since $\tau_\mathit{must} \models \tau_\mathit{can}$, we also know $\forall n \in \den{\tau_\mathit{must}}. \ \rhoOut(x)(n) \ne \bot$.
By definition, it follows $\pden{x}^n_{\rhoIn \cdot \rhoOut} = (\rhoIn \cdot \rhoOut)(x)(n) = \rhoOut(x)(n) \ne \bad$.

\bigskip\noindent\textbf{(4) Addition:}
Suppose $\Gamma \models e_1 : \tau$ and $\Gamma \models e_2 : \tau$. We have to show $\Gamma \models \stradd{e_1}{e_2} : \tau$.
Let $\rhoIn \in \mathit{Smap}(\Vin)$ and $\rhoOut \in \Lctx{\Gamma}_\rhoIn$. We have to show \[
  \forall n \in \den{\tau}_\rhoIn. \pden{\stradd{e_1}{e_2}}^n_{\rhoIn \cdot \rhoOut} \ne \bad
\]

Since $\Gamma \models e_1 : \tau$ and $\Gamma \models e_2 : \tau$, we know that
for all $n \in \den{\tau}_\rhoIn$, $\pden{e_1}^n_{\rhoIn \cdot \rhoOut} \ne \bad$ and $\pden{e_2}^n_{\rhoIn \cdot \rhoOut} \ne \bad$.
In turn, for all $n \in \den{\tau}_\rhoIn $ we have $\pden{\stradd{e_1}{e_2}}^n_{\rhoIn \cdot \rhoOut} = \pden{e_1}^n_{\rhoIn \cdot \rhoOut} +_\bad \pden{e_2}^n_{\rhoIn \cdot \rhoOut} \ne \bad$.

\bigskip\noindent\textbf{(5) Past accesses to output streams:}
Suppose $(x : \tau_\mathit{can}) \in \Gamma$, $\Gamma \models e : \tau_\mathit{must}$ and $\tau_\mathit{must} \models \tau_\mathit{can}$. We have to show $\Gamma \models \strlast{x}{e} : \tau_\mathit{must}$.
Let $\rhoIn \in \mathit{Smap}(\Vin)$ and $\rhoOut \in \Lctx{\Gamma}_\rhoIn$. It suffices to show \[
  \forall n \in \den{\tau_\mathit{must}}_\rhoIn. \pden{\strlast{x}{e}}^n_{\rhoIn \cdot \rhoOut} \ne \bad
\]

Let $n \in \den{\tau}_\mathit{must}$. Since $\Gamma \vDash e : \tau_\mathit{must}$,
we know $\pden{e}^n_{\rhoIn \cdot \rhoOut} \ne \bad$. Additionally, since $x : \tau_\mathit{can} \in \Gamma$ and $\rhoOut \in \Lctx{\Gamma}_\rhoIn$ we know that $x \in \mathit{dom}(\rhoOut)$. Therefore, by definition, $\pden{\strlast{x}{e}}^n_{\rhoIn \cdot \rhoOut} = \mathit{Prev}(\rhoOut(x), n, \pden{e}^n_{\rhoIn \cdot \rhoOut})$. Finally, we observe that for any $n$, any stream $w$ such that $w(n) \ne \bot$, and any value $v \ne \bad$, $\mathit{Prev}(w, n, v) \ne \bad$ (by definition of $\mathit{Prev}$).
Since $\tau_\mathit{must} \models \tau_\mathit{can}$ and $\rhoOut \in \Lctx{\Gamma}_\rhoIn$, clearly $\rhoOut(x)(n) \ne \bad$. Further, we know $\pden{e}^n_{\rhoIn \cdot \rhoOut} \ne \bad$. Consequently, $\pden{\strlast{x}{e}}^n_{\rhoIn \cdot \rhoOut} = \mathit{Prev}(\rhoOut(x), n, \pden{e}^n_{\rhoIn \cdot \rhoOut}) \ne \bad$.
\end{proof}

From the previous lemma we obtain the desired soundness theorem as a simple corollary.

\begin{corollary}
  $\Gamma \vdash e : \tau \implies \Gamma \vDash e : \tau$
\end{corollary}
\begin{proof}
  By immediate induction on the derivation $\Gamma \vdash e : \tau$, using items $(1)$ to $(7)$ from \Cref{lem:expression-compat} to handle each case.
\end{proof}

\headline{Soundness for Specifications}

Now that we established the soundness of the typing rules for expressions, we still need to prove that the typing rules for lists of equations are sound. The theorem we wish to prove is stated as follows.

\begin{theorem}[Soundness of Typing for Specifications]
  $\Gamma \vdash S \implies \Gamma \models S$
\end{theorem}

As for typing of expressions, we decompose this proof in two lemmas, each corresponding to the soundness of one typing rule for stream equations.

\begin{lemma}[Soundness of typing rules for equations]
  \label[plemma]{lem:equation-compat}
  The following two statements hold:
  \begin{enumerate}
  \item $\Gamma \vDash \epsilon$
  \item $x \notin \mathit{dom}(\Gamma) \implies \Gamma, x : \tau \vDash S \implies \Gamma \vDash e : \tau \implies \Gamma \vDash (\strdef{x}{\tau}{e}) \cdot S$
\end{enumerate}
\end{lemma}
\begin{proof}
  We prove the two implications separately.

  \bigskip\noindent\textbf{Proof of (1)}:
  Let $\rhoIn \in \mathit{Smap}(\Vin)$ and $\rho^1_\mathit{out} \in \Lctx{\Gamma}_\rhoIn$. We have to show the existence of some $\rho^2_\mathit{out} \in \mathit{Smap}(\Vout \setminus \mathit{dom}(\Gamma))$ such that $\rhoIn \cdot \rho^1_\mathit{out} \cdot \rho^2_\mathit{out} \in \den{\epsilon}$
  Since $\den{\epsilon} = \mathit{Smap}(\Vin) \cdot \mathit{Smap}(\Vout)$, it suffices to pick any $\rho^2_\mathit{out} \in \mathit{Smap}(\Vout \setminus \mathit{dom}(\Gamma))$.
  We pick $\rho^2_\mathit{out} = \lambda \_. \lambda \_. \bot$.

  \bigskip\noindent\textbf{Proof of (2)}:
  Let $x \notin \mathit{dom}(\Gamma)$, and suppose that $\Gamma, x : \tau \vDash S$ and $\Gamma \vDash e : \tau$. We have to show $\Gamma \vDash (\strdef{x}{\tau}{e}) \cdot S$.
  Let $\rhoIn \in \mathit{Smap}(\Vin)$, and $\rho^1_\mathit{out} \in \Lctx{\Gamma}_\rhoIn$. We need to show the existence of some $\rho^2_\mathit{out} \in \mathit{Smap}(\Vout \setminus \mathit{dom}(\Gamma))$ such that $
    \rhoIn \cdot \rho^1_\mathit{out} \cdot \rho^2_\mathit{out} \in \den{(\strdef{x}{\tau}{e}) \cdot S} = \den{\strdef{x}{\tau}{e}} \cap \den{S}
  $.
  Since $\Gamma \vDash e : \tau$ and $\rho^1_\mathit{out} \in \Lctx{\Gamma}_\rhoIn$, we know that $\forall n \in \den{\tau}_\rhoIn. \pden{e}^n_{\rhoIn \cdot \rho^1_\mathit{out}} \ne \bad$.
  Therefore, we can construct the following stream $w_e \in \mathit{Stream}$: \[
    w_e \triangleq \lambda n. \begin{cases}\begin{tabular}{ll}
      $\pden{e}^n_{\rhoIn \cdot \rho^1_\mathit{out}}$ &\textrm{if $n \in \den{\tau}$}\\
      $\bot$ &\textrm{otherwise}
    \end{tabular}\end{cases}
  \]
  Clearly $\rho^1_\mathit{out} \cdot (x \mapsto w_e) \in \Lctx{\Gamma, x : \tau}_\rhoIn$ and since $\Gamma, x : \tau \vDash S$, there exists $\rho^S_\mathit{out} \in \mathit{Smap}(\Vout \setminus \mathit{dom}(\Gamma, x : \tau))$ such that $
    \rhoIn \cdot (\rho^1_\mathit{out} \cdot (x \mapsto w_e)) \cdot{\rho^S_\mathit{out}} = \rhoIn \cdot \rho^1_\mathit{out} \cdot ((x \mapsto w_e) \cdot \rho^S_\mathit{out}) \in \den{S}
  $.
  Now we can pick $\rho^2_\mathit{out} \triangleq (x \mapsto w_e) \cdot \rho^S_\mathit{out}$. As stated above, $\rhoIn \cdot \rho^1_\mathit{out} \cdot \rho^2_\mathit{out} \in \den{S}$. It only remains to show that $\rhoIn \cdot \rho^1_\mathit{out} \cdot \rho^2_\mathit{out} \in \den{\strdef{x}{\tau}{e}}$. By definition of the semantics of equations, this means we have to prove $\forall n \in \den{\tau}_\rhoIn. (\rho^1_\mathit{out} \cdot \rho^2_\mathit{out})(x)(n) = \den{e}^n_{\rhoIn \cdot \rho^1_\mathit{out} \cdot \rho^2_\mathit{out}}$.
  First, let us observe that $(\rho^1_\mathit{out} \cdot \rho^2_\mathit{out})(x) = (\rho^1_\mathit{out} \cdot (x \mapsto w_e) \cdot \rho^S_\mathit{out})(x) = w_e$ since $x \notin \mathit{dom}(\rho^1_\mathit{out})$ and $x \notin \mathit{dom}(\rho^S_\mathit{out})$.
  Further, we already know that for all $n \in \den{\tau}_\rhoIn$, $\pden{e}^n_{\rhoIn \cdot \rho^1_\mathit{out}} \ne \bad$. Therefore, by \Cref{lem:pden-den}, it follows that for any $n \in \den{\tau}_\rhoIn$, we have $
    \pden{e}^n_{\rhoIn \cdot \rho^1_\mathit{out}} = \den{e}^n_{\rhoIn \cdot \rho^1_\mathit{out} \cdot \rho^2_\mathit{out}}$, which concludes our proof.
\end{proof}

From the two items of \Cref{lem:equation-compat}, we easily obtain the desired theorem as a corollary.

\begin{corollary}
  \label{sem-soundness}
  $\Gamma \vdash S \implies \Gamma \vDash S$
\end{corollary}
\begin{proof}
  By induction on the typing derivation $\Gamma \vdash S$.

  \smallskip\noindent\textbf{Base case}: Suppose $S = \epsilon$ and $\Gamma \vdash S$. By item (1) of \Cref{lem:equation-compat}, we know $\Gamma \vDash S$.

  \smallskip\noindent\textbf{Inductive case}: Suppose $x \notin \mathit{dom}(\Gamma)$, $\Gamma \vdash e : \tau$, and $\Gamma, x : \tau \vdash S$.
  By induction, we also have $\Gamma, x : \tau \vDash S$. Further, since $\Gamma \vdash e : \tau$, by \Cref{lem:expression-compat} we have $\Gamma \vDash e : \tau$.
  By item (2) of \Cref{lem:equation-compat}, it follows that $\Gamma \vDash \strdef{x}{e}{\tau} \cdot S$.
\end{proof}

Finally, type-safety simply follows from \Cref{sem-soundness} and \Cref{thm:safety}.

\begin{corollary}
  $\epsilon \vdash S \implies \mathit{Safe}(S)$
\end{corollary}
\begin{proof}
  Suppose $\epsilon \vdash S$.
  By \Cref{sem-soundness} we have $\epsilon \vDash S$. By \Cref{thm:safety} this implies $\mathit{Safe}(S)$.
\end{proof}

\subsection{Equation Ordering}
%\begin{figure}
%  \begin{subfigure}{0.48\textwidth}
%    \begin{lstlisting}
%    	input i: Int
%    	output x @i@ := y
%    	output y @i@ := i
%    \end{lstlisting}
%    \caption{A specification ill-typed due to equation ordering.}
%    \label{fig:ordering:wrong}
%\end{subfigure}
%\hfill
%\begin{subfigure}{0.48\textwidth}
%        \begin{lstlisting}
%    	input i: Int
%    	output y @i@ := i
%    	output x @i@ := y
%    \end{lstlisting}
%    \caption{A reordered well-typed specification.}
%    \label{fig:ordering:corrected}
%\end{subfigure}
%\caption{Specifications highlighting the role of equation ordering.}
%\label{fig:ordering}
%\end{figure}

The typing rules for equations in \Cref{sec:rules-v1} process specification equations in their given order.
Since the semantics of \CLola does not depend on this order, the type system may reject semantically valid specifications based on equation order alone, as demonstrated by the following two examples: 
\begin{center}
\begin{minipage}{0.35\textwidth}
	\begin{lstlisting}[label=fig:ordering:wrong,caption={A specification ill-typed due to equation ordering.}]
    	input i: Int
    	output x @i@ := y
    	output y @i@ := i
    \end{lstlisting}
\end{minipage}
\hspace{3em}
\begin{minipage}{0.35\textwidth}
	 \begin{lstlisting}[label=fig:ordering:corrected,caption={A reordered well-typed specification.}]
    	input i: Int
    	output y @i@ := i
    	output x @i@ := y
    \end{lstlisting}
\end{minipage}
\end{center}

The specifications in \Cref{fig:ordering:wrong} and \Cref{fig:ordering:corrected} differ only in the order of the stream equations for $x$ and $y$.
Both are semantically valid with respect to their annotated pacing types.
However, the type system in \Cref{sec:rules-v1} rejects the specification in \Cref{fig:ordering:wrong} because it type-checks the equation for $x$ first, where the expression directly accesses stream $y$.
The typing rule for direct access requires $y$'s type to be present in the type environment $\Gamma$, but $y$ is processed and therefore added to $\Gamma$ only after $x$.

A simple solution is to swap the order of $x$ and $y$, as shown in \Cref{fig:ordering:corrected}.
To generalize this, we introduce a specification-level rule: a specification is well-typed if some permutation of its equations is well-typed.

%We note that the type system we presented in the previous subsection is extremely
%sensitive to the ordering of the equations, even though in principle, the formal
%semantics is completely independent of the exact ordering.
%
%\todo[inline]{Show an example of correct list of equation that would not type check even tho it
%is semantically valid}.
%
%This problem can be fixed by adding a rule allowing to reorder equations: 
\[
  \inferrule[Permutation]{\exists S' \in \mathit{Permut}(S). \Gamma \vdash S'}{\Gamma \vdash S}
\]
Note that adding this rule preserves the soundness result, as the semantics of \CLola is independent of the order of equations.
\begin{lemma}
	For any specification $S$ it holds that:
	\[
		S' \in \mathit{Permut}(S) \implies \mathit{Safe}(S') \implies \mathit{Safe}(S) 
	\]
\end{lemma}
  \section{Pacing Types With Self References}\label{sec:pacing_types_self}
This section weakens the type system in \Cref{sec:pacing_types} by allowing streams to access their own past values.
This enables specifications such as the running average computation in \Cref{fig:running_average}.
\begin{figure}[ht!]
\begin{tabular}{c}
\begin{lstlisting}
	input i: Int
	output count @i@ := count.prev(or: 0) + 1
	output sum @i@ := sum.prev(or: 0) + i
	output average @i@ := sum / count
\end{lstlisting}
\end{tabular}
\caption{A \CLola specification computing the running average of input $i$}
\label{fig:running_average}
\end{figure}
Intuitively, such specifications are safe as long as a stream depends (transitively) only on its past values.
At runtime, these values have already been computed and are available for the stream's next value computation.
However, the pacing type system in \Cref{sec:pacing_types} rejects such specifications because the \textsc{PrevOut} rule requires the accessed stream to be in the typing environment, while the \textsc{Eq} rule types the expression under an environment that explicitly excludes the stream itself.
The extended type system introduced in this section permits such specifications by explicitly tracking the stream to which an expression belongs.

\subsection{Typing Rules Allowing Self References}
As in \Cref{sec:pacing_types}, the extended typing rules operate on two levels: specifications (i.e., lists of equations) and expressions.
At the specification level, the typing rules in \Cref{fig:sr-rules:eqs} type expressions using a modified typing judgment ($\vdash^{\color{red} x}$), which records the output stream associated with the expression.
At the expression level, the rules in \Cref{fig:sr-rules:exprs} ensure that an output stream does not access itself directly, i.e., the accessed stream must differ from $x$.
Previously, this restriction was an implicit consequence of the $x \not \in \mathit{dom}(\Gamma)$ requirement in the \textsc{Eq} rule and the requirement that $(x: \tau_\mathit{can}) \in \Gamma$ for stream access expressions in \Cref{sec:pacing_types}.
An exception is the \textsc{Self} rule, which enables the desired extension.
It allows typing a previous access while explicitly enforcing that the accessed stream is the stream itself ($x = y$).
Yet, unlike the \textsc{Sr-PrevOut} rule, it does not require the accessed stream to be present in the typing environment $\Gamma$.

\begin{figure}[ht!]
\begin{mathpar}
  \inferrule[Sr-Empty]{ }{\Gamma \vdash \epsilon} \and \inferrule[Sr-Eq]{ x \notin \mathit{dom}(\Gamma)\\\Gamma \vdash^x e : \tau\\ \Gamma, x : \tau \vdash \varphi }{\Gamma \vdash (\strdef{x}{\tau}{e}) \cdot \varphi }\\
\end{mathpar}
\caption{The \CLola pacing type rules for stream equations allowing self-references}
\label{fig:sr-rules:eqs}
\end{figure}

%Note that here, when type checking an expression, we remember the name of the stream this expression is a part of by annotating the judgement ($\vdash^{\color{red} x}$). Importantly, in these new typing judgements, $x$ is necessarily an output variable (i.e., $x \in \Vout$).

%\headline{Typing Expressions}
\begin{figure}[ht!]
\begin{mathpar}
  \inferrule[Sr-Const]{v \in \mathbb{Z}}{\Gamma \vdash^x v : \tau}\and
  \inferrule[Sr-BinOp]{\Gamma \vdash^x e_1 : \tau_\mathit{must} \\ \Gamma \vdash^x e_2 : \tau_\mathit{must} }{\Gamma \vdash^x \stradd{e_1}{e_2} : \tau_\mathit{must}}\\
  \inferrule[Sr-DirectOut]{y \in \mathbb{V}_\texttt{out}\\ x \ne y \\(y : \tau_\mathit{can}) \in \Gamma\\ \tau_\mathit{must} \models \tau_\mathit{can}}{\Gamma \vdash^x y : \tau_\mathit{must}}\and
  \inferrule[Sr-DirectIn]{y \in \mathbb{V}_\texttt{in} \\ \tau_\mathit{must} \models y }{\Gamma \vdash^x y : \tau_\mathit{must}}\\
  \inferrule[Sr-HoldOut]{x \ne y\\ y \in \mathit{dom}(\Gamma)\\\Gamma \vdash^x e : \tau_\mathit{must} }{\Gamma \vdash^x \strhold{y}{e} : \tau_\mathit{must}}\and
    \inferrule[Sr-HoldIn]{y \in \Vin\\ \Gamma \vdash^x e : \tau_\mathit{must} }{\Gamma \vdash^x \strhold{y}{e} : \tau_\mathit{must}}\\
  \inferrule[Sr-PrevOut]{x \ne y\\(y : \tau_\mathit{can}) \in \Gamma \\ \Gamma \vdash^x e : \tau_\mathit{must} \\ \tau_\mathit{must} \models \tau_\mathit{can}\\ }{\Gamma \vdash^x \strlast{y}{e} : \tau_\mathit{must}}\\
  \inferrule[Sr-PrevIn]{y \in \Vin \\ \Gamma \vdash^x e : \tau_\mathit{must} \\ \tau_\mathit{must} \models y\\ }{\Gamma \vdash^x \strlast{y}{e} : \tau_\mathit{must}}\qquad
  \inferrule[Self]{x = y \\ \Gamma \vdash^x e : \tau_\mathit{must}}{\Gamma \vdash^x \strlast{y}{e} : \tau_\mathit{must}}
\end{mathpar}
\caption{The \CLola pacing type rules for stream expressions allowing self-references}
\label{fig:sr-rules:exprs}
\end{figure}

\subsection{Updated Interpretation}

Consider the single equation $\strdef{x}{\syn{$\top$}}{\strlast{x}{c}}$.
This equation has a solution (map $x$ to $w_x = c^\omega$). As expected, it is accepted by the new type system, as demonstrated by the following derivation: \begin{mathpar}
  \inferrule*[left=Sr-Eq]{
    \inferrule*[left=Sr-Self]{
      \inferrule*[left=Sr-Const]{ }{\emptyset \vdash c : \syn{$\top$}}
    }{\emptyset \vdash \strlast{x}{c} : \syn{$\top$}}
  }{\emptyset \vdash \strdef{x}{\top}{\strlast{x}{c}}}
\end{mathpar}

However, this equation is not semantically well-typed, and we cannot prove $\emptyset \vDash \strlast{x}{c} : \syn{$\top$}$.
Indeed, for any $\rhoIn$, the only $\rhoOut$ in $\Lctx{\emptyset}_\rhoIn$ is $\emptyset$
and $\pden{\strlast{x}{c}}^n_{\rhoIn \cdot \emptyset} = \bad$ for all $n$.
More generally, in our current interpretation of pacing types,
an equation of the form $\strdef{x}{\tau}{e}$ is considered to be safe in the environment $\Gamma$ under the condition that $e$ can be safely evaluated using only streams in $\Gamma$, thus forbidding any form of self-references.
Nonetheless, the updated type system still ensures the safety of specifications.
To prove this, we update the semantic interpretation of pacing types.

\headline{Semantics of Expressions Under Partial Maps}

\newcommand{\mpden}[1]{\widetilde{\den{#1}}}

To define an interpretation of pacing types justifying the correctness of the new rule \textsc{Sr-Self}, we extend the partial semantics of expressions with an additionally \emph{named memory cell}.
Conceptually, this memory cell represents the last defined value of a stream
$x$ (where $x \in \Vout$).
As for $\pden{-}$, this semantics is defined for partial stream maps.

\begin{definition}[Semantics of expressions with named memory cell]
  \begin{align*}
    \mpden{c}^{n,x=v}_\rho &\triangleq c\\
    \mpden{\stradd{e_1}{e_2}}^{n,x=v}_\rho &\triangleq \mpden{e_1}^{n,x=v}_\rho +_{\bad} \mpden{e_2}^{n,x=v}_\rho\\
    \mpden{y}^{n,x=v}_\rho &\triangleq \begin{cases}
    	\begin{tabular}{ll}
    		$\pden{y}^n_\rho$ &if $y \ne x$\\
    		$\bad$ &if $y = x$\\
    	\end{tabular}
    \end{cases}\\
    \mpden{\strlast{y}{e}}^{n,x=v}_\rho &\triangleq \begin{cases}
    	\begin{tabular}{ll}
    		$\prev{\rho(y)}{n}{\mpden{e}^{n,x=v}_\rho}$ &if $y \ne x \land y \in \mathit{dom}(\rho)$\\
    		$\bad$ &if $y \ne x \land y \not\in \mathit{dom}(\rho)$\\
    		$v$ &if $y = x \land n > 0 \land v \ne \bad$\\
    		$\mpden{e}^{n,x=v}_\rho$ &if $y = x \land (n = 0 \lor v = \bad)$\\
    	\end{tabular}
    \end{cases}\\
    \mpden{\strhold{x}{e}}^{n,x=v}_\rho &\triangleq \begin{cases}
    	\begin{tabular}{ll}
    		$\hold{\rho(y)}{n}{\mpden{e}^{n,x=v}_\rho}$ &if $y \ne x \land y \in \mathit{dom}(\rho)$\\
    		$\bad$ &if $y = x \lor (y \ne x \land y \not\in \mathit{dom}(\rho))$\\
    	\end{tabular}
    \end{cases}
  \end{align*}
\end{definition}

We note that for any expression $e$ and any stream $x$ and any partial map $\rho$ that contains sufficiently many streams for $e$ to be safely evaluated, then $\den{e}$ and $\mpden{e}$ coincide:
\newcommand{\mlast}[2]{\widetilde{\mathit{Last}}(#1,#2)}
\begin{lemma}
  \label[plemma]{lem:boring2}
  Let $x \in \Vout$, $Y \subseteq \Vout \setminus \{x\}$, and $\rho \in \mathit{Smap}(\Vin \cup Y)$ and $\rho' \in \mathit{Smap}(\Vout \setminus Y \setminus \{ x \})$. Then for any stream $w \in \mathit{Stream}$ and any $n \in \mathbb{N}$, we have \[
    w(n) \ne \bot \implies \mpden{e}^{n, x = \mlast{w}{n-1}}_\rho \ne \bad \implies \mpden{e}^{n, x = \mlast{w}{n-1}}_\rho = \den{e}_{\rho \cdot (x \mapsto w) \cdot \rho'}
  \]
  Where $\mlast{w}{n}$ is defined as follows:
  \[
    \mlast{w}{n} \triangleq \begin{cases}
    	\begin{tabular}{ll}
    		$\last{w}{n}$ &if $n \geq 0$\\
    		$\bad$ &if $n < 0$\\
    	\end{tabular}
    \end{cases}
  \]
\end{lemma}

\begin{proof}
  By induction on $e$. The difficult case is the case of an past access $\strlast{y}{e}$.

  \bigskip\noindent\textbf{When $x \ne y$}: We have \begin{align*}
    \mpden{\strlast{y}{e}}^{n,x=\mlast{w}{n-1}}_\rho
    &= \begin{cases}
    \begin{tabular}{ll}
      $\prev{\rho(y)}{n}{\mpden{e}^{n,x=\mlast{w}{n-1}}_\rho}$ &if $y \in \mathit{dom}(\rho)$\\
      $\bad$ &if $y \not\in \mathit{dom}(\rho)$
    \end{tabular}
    \end{cases}\\
    \textrm{\textit{\color{gray}By induction}}\qquad
    &= \begin{cases}
      \begin{tabular}{ll}
        $\prev{\rho(y)}{n}{\den{e}^{n}_{\rho \cdot (x \mapsto w) \cdot \rho'}}$ &if $y \in \mathit{dom}(\rho)$\\
        $\bad$ &if $y \not\in \mathit{dom}(\rho)$
      \end{tabular}
    \end{cases}\\
    \textrm{\textit{\color{gray}By def. of $\pden{-}$}}\qquad
    &= \pden{e}^n_{\rho \cdot (x \mapsto w) \cdot \rho'}\\
    \textrm{\textit{\color{gray}By \Cref{lem:pden-den}}}\qquad
    &= \den{e}^n_{\rho \cdot (x \mapsto w) \cdot \rho'}
  \end{align*}

  \noindent\textbf{When $x = y$}:

  \begin{align*}
    \mpden{\strlast{x}{e}}^{n,x=\mlast{w}{n-1}}_\rho
    &=\begin{cases}\begin{tabular}{ll}
      $\mlast{w}{n-1}$ &if $n > 0 \land \mlast{w}{n-1} \ne \bad$\\
      $\mpden{e}^{n,x=\mlast{w}{n-1}}_\rho$ &if $n = 0 \lor \mlast{w}{n-1} = \bad$
    \end{tabular}
    \end{cases}\\
    \textrm{\textit{\color{gray}By def. of $\mlast{-}{-}$ for $n \ge 0$}}\qquad
    &= \begin{cases}\begin{tabular}{ll}
      $\last{w}{n-1}$ &if $n > 0 \land \last{w}{n-1} \ne \bad$\\
      $\mpden{e}^{n,x=\mlast{w}{n-1}}_\rho$ &if $n = 0 \lor \last{w}{n-1} = \bad$
    \end{tabular}
    \end{cases}\\
    \textrm{\textit{\color{gray}By induction}}\qquad
    &= \begin{cases}\begin{tabular}{ll}
      $\last{w}{n-1}$ &if $n > 0 \land \last{w}{n-1} \ne \bad$\\
      $\den{e}^n_{\rho \cdot (x \mapsto \rho) \cdot \rho'}$ &if $n = 0 \lor \last{w}{n-1} = \bad$
    \end{tabular}
    \end{cases}\\
    \textrm{\textit{\color{gray}By definition of $\mathit{Prev}$ if $w(n) \ne \bot$}}\qquad
    &= \mathit{Prev}(w, n, \den{e}^n_{\rho \cdot (x \mapsto \rho) \cdot \rho'})\\
    \textrm{\textit{\color{gray}By definition of $\den{-}^n$}}\qquad
    &= \den{\strlast{x}{e}}^n_{\rho \cdot (x \mapsto w) \cdot \rho'}
  \end{align*}
\end{proof}

\headline{Updated logical relation}

Building on the new semantics of expressions, we give an updated interpretation of pacing types.
The relations $\Lsafe{-}, \Lctx{-}, \Lstr{-}$ are left unchanged. However, we redefine the relation $\Lequ{-}$ as follows \begin{align*}
  \Lequ{\Gamma \mid x : \tau}_\rhoIn \triangleq \{ \ \forall \rhoOut \in \Lctx{\Gamma}_\rhoIn.
  \forall n \in \den{\tau}_\rhoIn. \forall v \in \mathbb{Z} \cup \{ \bad \}. \ \mpden{e}^{n, x = v}_{\rhoIn \cdot \rhoOut} \neq \bad \ 
  \}
\end{align*}

Note that this time, $\Lequ{-}$ takes a variable $x$ as an argument in addition to the context $\Gamma$ and the annotation $\tau$.
We use the notation $\Gamma \models^x e : \tau \triangleq \forall \rhoIn. e \in \Lequ{\Gamma \mid x : \tau}_\rhoIn$ for semantic typing of expressions using the updated interpretation.

\subsection{Soundness Proof V2}

\begin{lemma}[Soundness of typing rules for expressions]
    The following implications hold:
    \begin{enumerate}
    \item $v \in \mathbb{Z} \implies \Gamma \vDash^x v : \tau$
    \item $x \in \mathbb{V}_\texttt{in} \implies \Gamma \vDash^x x : \tau$
    \item $x \in \mathbb{V}_\texttt{out} \implies (x : \tau_\mathit{can}) \in \Gamma \implies \tau_\mathit{must} \models \tau_\mathit{can} \implies \Gamma \vDash^x x : \tau_\mathit{must}$
    \item $\Gamma \vDash^x e_1 : \tau \implies \Gamma \vDash^x e_2 : \tau \implies \Gamma \vDash^x \stradd{e_1}{e_2} : \tau $
    \item $y \in \Vin \implies \vDash^x e : \tau \implies \tau \models y \implies \Gamma \vDash^x \strlast{y}{e} : \tau$
    \item $\Gamma \vDash^x e : \tau_\mathit{must} \implies \Gamma \vDash^x \strlast{x}{e}: \tau_\mathit{must} $
    \item $x \neq y \implies \Gamma \vDash^x e : \tau_\mathit{must} \implies \tau_\mathit{must} \models \tau_\mathit{can} \implies \Gamma \vDash^x \strlast{y}{e}: \tau_\mathit{must}$
    \item $y \in \Vin \implies \Gamma \vDash^x e : \tau \implies \Gamma \models \strhold{y}{e} : \tau$
    \item $x \ne y \implies y \in \mathit{dom}(\Gamma) \implies \Gamma \vDash^x (e : \tau) \implies \Gamma \models \strhold{y}{e} : \tau$
  \end{enumerate}
\end{lemma}
\begin{proof}
  As for pacing types without self-references, we only focus on a selection of challenging cases.

  \bigskip\noindent\textbf{(3) Direct access to an output $x \ne y$:}
  Let $y$ an output variable such that $x \ne y$ and $(y : \tau_\mathit{can}) \in \Gamma$.
  Further, suppose that $\tau_\mathit{must} \models \tau_\mathit{can}$.
  We have to show $\Gamma \models^x y : \tau_\mathit{must}$.
  Let $\rhoIn \in \mathit{Smap}(\Vin)$, and $\rhoOut \in \Lctx{\Gamma}_\rhoIn$, it suffices to show \[
    \forall n \in \den{\tau_\mathit{must}}. \forall v \in \mathbb{Z} \cup \{ \bad \}. \ \mpden{y}^{n, x = v}_{\rhoIn \cdot \rhoOut} = \bad
  \]
  Let $n \in \den{\tau_\mathit{must}}$ and $v \in \mathbb{Z} \cup \{ \bad \}$. Since $y \ne x$, by definition we have $\mpden{y}^{n, x = v}_{\rhoIn \cdot \rhoOut}
    = \pden{y}^n_{\rhoIn \cdot \rhoOut}$.
  Further, from $\rhoOut \in \Lctx{\Gamma}_\rhoIn$, $(y : \tau_\mathit{can}) \in \Gamma$, $n \in \den{\tau_\mathit{must}}$ and $\tau_\mathit{must} \models \tau_\mathit{can}$ we can easily deduce $\rhoOut(y)(n) \ne \bot$ and therefore $\pden{y}^n_{\rhoIn \cdot \rhoOut} = \rhoOut(y)(n) \in \mathbb{Z}$.
  It follows that $\mpden{y}^{n, x = v}_{\rhoIn \cdot \rhoOut} = \rhoOut(y)(n) \ne \bad$

  \bigskip\noindent\textbf{(4) Addition:}
  Suppose $\Gamma \vDash^x e_1 : \tau$ and $\Gamma \vDash^x e_2 : \tau$. We have to show $\Gamma \vDash^x \stradd{e_1}{e_2} : \tau$.
  Let $\rhoIn \in \mathit{Smap}(\Vin)$ and $\rhoOut \in \Lctx{\Gamma}_\rhoIn$. We have to show \[
    \forall n \in \den{\tau_\mathit{must}}. \forall v \in \mathbb{Z} \cup \{ \bad \}. \ \mpden{\stradd{e_1}{e_2}}^{n, x = v}_{\rhoIn \cdot \rhoOut} \ne \bad
  \]
  Let $n \in \den{\tau_\mathit{must}}$ and $v \in \mathbb{Z} \cup \{ \bad \}$. Since $\Gamma \vDash^x e_1$ and $\Gamma \vDash^x e_2$, we know $\mpden{e_2}^{n, x = v}_{\rhoIn \cdot \rhoOut} \ne \bad$ and $\mpden{e_1}^{n, x = v}_{\rhoIn \cdot \rhoOut} \ne \bad$. It immediately follows that $\mpden{\stradd{e_1}{e_2}}^{n, x = v}_{\rhoIn \cdot \rhoOut} = \mpden{e_1}^{n, x = v}_{\rhoIn \cdot \rhoOut} +_\bad \mpden{e_2}^{n, x = v}_{\rhoIn \cdot \rhoOut} \ne \bad$.

  \bigskip\noindent\textbf{(6) Past access to $x$:} Suppose $\Gamma \vDash^x e : \tau$. We have to show $\Gamma \vDash^x \strlast{x}{e} : \tau$.
  Let $\rhoIn \in \mathit{Smap}(\Vin)$ and $\rhoOut \in \Lctx{\Gamma}_\rhoIn$. It suffices to show \[
    \forall n \in \den{\tau}. \forall v \in \mathbb{Z} \cup \{ \bad \}. \ \mpden{\strlast{y}{e}}^{n, x = v}_{\rhoIn \cdot \rhoOut} \ne \bad
  \]
  Let $n \in \den{\tau}$ and $v \in \mathbb{Z} \cup \{ \bad \}$. By definition, $\mpden{\strlast{x}{e}}^{n, x=v}_{\rhoIn \cdot \rhoOut}$ is equal to $v$ if $v \ne \bad$. In that case, we are done.
  Otherwise, $\mpden{\strlast{x}{e}}^{n, x=v}_{\rhoIn \cdot \rhoOut} = \mpden{e}^{n, x=v}_{\rhoIn \cdot \rhoOut}$.
  Since $\Gamma \vDash^x e : \tau$, we know $\mpden{e}^{n, x=v}_{\rhoIn \cdot \rhoOut} \ne \bad$.

  \bigskip\noindent\textbf{(7) Past access to an output $x \ne y$:}
  Suppose $y : \tau_\mathit{can} \in \Gamma$, $\Gamma \vDash^x e : \tau_\mathit{must}$ and $\tau_\mathit{must} \models \tau_\mathit{can}$. We have to show $\Gamma \models \strlast{y}{e} : \tau_\mathit{must}$.
  Let $\rhoIn \in \mathit{Smap}(\Vin)$ and $\rhoOut \in \Lctx{\Gamma}_\rhoIn$. It suffices to show \[
    \forall n \in \den{\tau_\mathit{must}}. \forall v \in \mathbb{Z} \cup \{ \bad \}. \ \mpden{\strlast{y}{e}}^{n, x = v}_{\rhoIn \cdot \rhoOut} \ne \bad
  \]
  Let $n \in \den{\tau_\mathit{must}}$ and $v \in \mathbb{Z} \cup \{ \bad \}$. Since $y \ne x$ and $y \in \mathit{dom}(\rhoOut)$, by definition $\mpden{\strlast{y}{e}}^{n, x=v}_{\rhoIn \cdot \rhoOut} = \mathit{Prev}(\rhoOut(y), n, \mpden{e}^{n, x=v}_{\rhoIn \cdot \rhoOut})$.
  As observed in \Cref{lem:expression-compat}, $\mathit{Prev}(\rhoOut(y), n, \mpden{e}^{n, x=v}_{\rhoIn \cdot \rhoOut})$ can only be equal to $\bad$ if (i) $\mpden{e}^{n, x = v}_{\rhoIn \cdot \rhoOut} = \bad$ or if (ii) $\rhoOut(y)(n) = \bot$.
  Since we assumed $\Gamma \vDash^x e : \tau_\mathit{must}$, clearly $\mpden{e}^{n, x = v}_{\rhoIn \cdot \rhoOut} \ne \bad$, contradicting (i).
  Further, since $n \in \den{\tau_\mathit{must}}_\rhoIn$ and $\tau_\mathit{must} \models \tau_\mathit{can}$, we can deduce that $n \in \den{\tau_\mathit{can}}_\rhoIn$. Then, since $(y : \tau_\mathit{can}) \in \Gamma$ and $\rhoOut \in \Lctx{\Gamma}_\rhoIn$, it follows that $\rhoOut(y)(n) \ne \bot$, contradicting (ii).
\end{proof}

\begin{lemma}[Soundness of typing rules for equations]\ \begin{enumerate}
  \item $\Gamma \vDash \epsilon$
  \item $x \notin \mathit{dom}(\Gamma) \implies \Gamma \vDash^x e : \tau \implies \Gamma, x : \tau \vDash S \implies \Gamma \models (\strdef{x}{\tau}{e}) \cdot S$
\end{enumerate}
\end{lemma}

\begin{proof}
  The proof of the first item is identical to the corresponding proof in \Cref{lem:equation-compat}.
  We focus on the second item. Suppose $x \notin \mathit{dom}(\Gamma)$, $\Gamma \vDash^x e : \tau$, and $\Gamma, x : \tau \vDash S$. We have to show $\Gamma \vDash \strdef{x}{\tau}{e} \cdot S$.
  Let $\rhoIn \in \mathit{Smap}(\Vin)$, and $\rho^1_\mathit{out} \in \Lctx{\Gamma}_\rhoIn$. We have to prove the existence of $\rho^2_\mathit{out} \in \mathit{dom}(\Vout \setminus \mathit{dom}(\Gamma))$ such that $
    \rhoIn \cdot \rho^1_\mathit{out} \cdot \rho^2_\mathit{out} \in \den{\strdef{x}{\tau}{e}}\cap \den{S}
  $.
  Since $\Gamma \vDash^x e : \tau$, by definition we know \[
    \forall n \in \den{\tau}_\rhoIn. \forall v_x \in \mathbb{Z} \cup \{ \bad \}. \mpden{e}^{n, x=v} \ne \bad
  \]

  We construct a stream for $e$ \emph{recursively} as follows: \begin{align*}
    w_e(n) \triangleq \begin{cases}\begin{tabular}{ll}
      $\mpden{e}^{n, x=\mlast{w_e(0) \cdot \ldots \cdot w_e(n - 1) \cdot 0^\omega}{n - 1}}$ & if $n \in \den{\tau}_\rhoIn$\\
      $\bot$ & if $n \notin \den{\tau}_\rhoIn$
    \end{tabular}\end{cases}
  \end{align*}

  Note that $w_e(n)$ is well defined because $w_e(n)$ only depends on $w_e$ at time points $n' < n$. Further, by the previous assumption, $w_e(n) \ne \bad$ for all $n \in \den{\tau}_\rhoIn$ and
  therefore, $w_e \in \Lstr{\tau}_\rhoIn$.
  Further, since $\mathit{Last}(w, n)$ only depends on the positions of $w$ up to $n$ for any stream $w$,
  it is not difficult to see that $\mlast{w_e(0)\cdot \ldots \cdot w_e(n - 1) \cdot 0^w}{n - 1} = \mlast{w_e}{n - 1}$.
  Therefore, for any $n \in \den{\tau}_\rhoIn$, we have $w_e(n) = \mpden{e}^{n, x=\mlast{w_e}{n - 1}}$.
  Now by applying $\Gamma \vDash S$ to the stream map $\rho^1_\mathit{out} \cdot (x \mapsto w_e)$, we obtain a map $\rho^S_\mathit{out}$ compatible with the equations in $S$.
  By the same reasoning as in \Cref{lem:equation-compat}, and by using \Cref{lem:boring2},
  we easily conclude that $(x \mapsto w_e) \cdot \rho^2_\mathit{out}$ is also compatible with $\strdef{x}{\tau}{e}$.
\end{proof}
   \section{Related Work and Discussion}\label{sec:related}

\paragraph*{Runtime Verification and Stream-based Monitoring}
Runtime verification is a mature idea which was first explored in the late 90s~\cite{DBLP:conf/ecrts/KimVBKLS99}.
Traditionally, specifications were expressed in temporal logics such as LTL~\cite{bauer2011runtime} or MTL~\cite{DBLP:conf/cav/BasinKZ17}.
The systems under observation range from Markov Decision Processes~\cite{DBLP:conf/cav/JungesTS20,DBLP:conf/cav/HenzingerKKM23} to cyber-physical systems~\cite{volostream} and software~\cite{DBLP:conf/icse/JinMLR12}.
Expressive specification techniques such as stream-based specification languages have been explored to capture the precise behavior of such systems.
Stream-based monitoring approaches such as Lola~\cite{d2005lola}, its successor \textsc{RTLola}~\cite{DBLP:conf/fm/BaumeisterFKS24}, Tessla~\cite{DBLP:conf/rv/KallwiesLSSTW22} or Striver~\cite{DBLP:conf/fm/GorostiagaS21} offer a fully-featured programming language.

\paragraph*{Synchronous Programming and Clocks}
Another programming paradigm that resembles stream-based monitoring in the style of RTLola is synchronous programming as implemented by programming languages such as LUSTRE~\cite{DBLP:journals/pieee/HalbwachsCRP91}, ESTEREL~\cite{DBLP:conf/concur/BerryC84} or Zélus~\cite{DBLP:conf/emsoft/BenvenisteBCP11}.
While stream-based monitoring languages have evolved independently of synchronous languages, there are overlapping problems and solutions in both fields of research.
For example, the fundamental problem of ensuring that a stream value exists when needed for further computations is present in both paradigms.
Similar to the pacing-types presented in this paper, synchronous languages employ various notions of \emph{clocks}~\cite{DBLP:conf/fpca/GautierG87,DBLP:conf/emsoft/ColacoP03,DBLP:conf/mpc/MandelPP10} and corresponding static analysis to solve timing issues.
Likewise,~\cite{DBLP:conf/hybrid/BenvenisteBCPP14} presents a type system to identify \emph{non-causal} equations in synchronous programs.
The property of \emph{causality} is related to the existence and uniqueness of models in stream-based monitoring.
Aside from mitigating timing issues, type systems enforcing other kinds of properties have been developed for synchronous programming languages. 
For example, an advanced refinement type system for the synchronous language MARVeLus has been explored in~\cite{DBLP:journals/pacmpl/ChenMABJSWZJ24}. 
This type system reasons about the temporal behavior of stream values based on LTL~\cite{DBLP:conf/focs/Pnueli77} specifications.

%While exploring the similarities and differences between synchronous programming and stream-based monitoring languages more deeply is intriguing it is out of the scope for this paper and is left for future work.
%{\color{red}However, the type system presented in this paper is unique for stream-based languages and together with the presented proofs of soundness using logical relations is a novel and highly beneficial contribution with practical relevance~\cite{DBLP:conf/cav/FaymonvilleFSSS19,DBLP:conf/cav/BaumeisterFSST20,DBLP:journals/tecs/BaumeisterFST19,volostream}}.\todo{note for Arthur: round the corners}

\paragraph*{Towards a Type-guided Interpreter for \CLola}

It is interesting to observe that, although we never introduced an operational semantics for \CLola
in this paper,
the soundness proof of the pacing type system provides a sketch of an \emph{interpreter}, which would safely compute a model for any well-typed list of equations.
Indeed, the proof constructs a model starting from the empty map,
progressively extending it by adding new streams satisfying each equation.
Since the type system enforces that each stream equation only depends on the previous equations
(or contains self-references to its own past), this \emph{strategy} never fails and gives
a systematic method to compute models of specifications.
This observation is reminiscent of an observation made by Kokke, Siek and Wadler in \cite{KOKKE2020102440}.
They observe that the computational content of a constructive proof of type-safety with respect to an operational semantics
(typically, a proof by \emph{progress and preservation}) 
constitutes a correct-by-construction interpreter.
This observation has been successfully applied, for example, to extract a correct interpreter for WASM \cite{10.1145/3704858}.

Even though the \CLola semantics presented in this paper is not formulated in an operational style, we still provided a notion of type-safety: any well-typed list of equations is guaranteed to have at least one model.
As an intermediate step, we provided an alternative semantics for stream expressions ($\mpden{-}$),
and constructed a semantic interpretation of pacing types against this semantics.
Conceptually, this alternative semantics is reminiscent of a more operational, state-based variant
of the original semantics of \CLola.
Under this interpretation, the soundness of the type system with respect to the semantic interpretation of pacing types can be viewed as a proof of correspondence between the base semantics and a more operational one.
Further, one could follow the approach of \cite{KOKKE2020102440} and \cite{10.1145/3704858} and try to extract an executable interpreter from the type-safety proof.
Another idea would be to exploit the alternative semantics of expressions to \emph{calculate} \cite{bahr15jfp} a correct compiler that translates well-typed \CLola specifications into executable imperative code computing a model.
We leave these ideas as future work.

\section{Conclusion}\label{sec:conclusion}
In this paper, we formalized pacing annotations, a simple yet expressive mechanism to specify
data synchronization policies in stream-based monitoring.
We formalized the essence of pacing annotations by presenting \CLola, a core language to specify monitors for asynchronous streams.
We equipped \CLola with a formal semantics that interprets annotated stream equations as relations between input and output streams.
We showed that, while a useful feature, pacing annotations can introduce inconsistencies in stream-based monitors.
Semantically, these inconsistencies can be understood as the non-existence of a model for a given set of annotated equations.
We addressed this issue by viewing pacing annotations as \emph{types}, and equipping \CLola with a type system that rejects specifications with contradictory pacing annotations.
We proved the soundness of this type system against the formal semantics of \CLola.
Importantly, since our soundness proof is developed with respect to the \emph{relational} semantics of \CLola, the soundness of our type system is independent of any concrete implementation of \CLola: the presented pacing types enforce the existence of a model; how to compute such a model is irrelevant.
Nonetheless, the constructive nature of our proof reveals an evaluation strategy for well-typed specifications.
  
  \bibliographystyle{ACM-Reference-Format}
  \bibliography{references.bib}

\end{document}